\documentclass[letterpaper,11pt]{article}
\usepackage{setspace}
\usepackage{amsmath}
\DeclareMathOperator*{\argmax}{arg\,max}

\usepackage{mathtools}

\usepackage[T1]{fontenc}
\usepackage[margin=1in]{geometry}
\usepackage[numbers, semicolon]{natbib}
\setcitestyle{authoryear,open={(},close={)},aysep={}} 
\usepackage{amsfonts}
\usepackage{float}
\usepackage{graphicx}
\usepackage{tikz}
\usetikzlibrary{arrows,shapes,snakes,automata,backgrounds,petri,fit, positioning,patterns,decorations.pathreplacing}
\usepackage{mathrsfs}
\usepackage{caption}
\usepackage{subcaption}
\usepackage[titles,subfigure]{tocloft}
\usepackage[ruled,procnumbered]{algorithm2e}
\usepackage{xcolor}
\usepackage{authblk}
\usepackage{bm}
\usepackage{hyperref}
\newcommand{\simiid}{\stackrel{\textnormal{i.i.d.}}{\sim}}
\usepackage{amsthm}

\newtheorem{theorem}{Theorem}
\newtheorem{definition}{Definition}
\newtheorem{proposition}{Proposition}

\newtheorem{lemma}{Lemma}
\newtheorem{remark}{Remark}

\newcommand{\rr}{\mathbb R}
\newcommand{\p}{\mathbb P}

\newcommand{\di}[1]{\mathop{d#1}}
\DeclareMathOperator{\Var}{Var}

\DeclareMathOperator{\Unif}{Unif}
\DeclareMathOperator{\diag}{diag}

\newcommand{\e}{\mathbb E}

\newcommand{\T}{\top}

\title{Stratification and Optimal  Resampling for  Sequential  Monte Carlo }

\author[1]{
Yichao Li\thanks{These authors contributed equally and are listed in alphabetical order.}}
\newcommand\CoAuthorMark{\footnotemark[\arabic{footnote}]}
\author[2]{Wenshuo Wang\protect\CoAuthorMark}
\author[1]{Ke Deng}
\author[2]{Jun S Liu}
\date{\today}
\affil[1]{Center for Statistical Science, Tsinghua University, Beijing 100084, China}
\affil[2]{Department of Statistics, Harvard University, Cambridge, MA 02138}
\date{\today}

\begin{document}
\maketitle
\begin{abstract}
Sequential Monte Carlo, also known as particle filters, has been widely accepted as a powerful computational tool for making inference with dynamical systems.  A key step in sequential Monte Carlo is resampling, which plays the  role of steering the algorithm towards the future dynamics. Several strategies have been proposed and used in practice, including multinomial resampling, residual resampling, optimal resampling, stratified resampling, and optimal transport resampling. We show that, in  one dimensional cases, optimal transport resampling is equivalent to stratified resampling on the sorted particles, and they both minimize the resampling variance as well as the expected squared energy distance between the original and resampled empirical distributions; in  multidimensional cases, the variance of stratified resampling after sorting particles using Hilbert curve in $\rr^d$ is $O(m^{-(1+2/d)})$, an improved rate compared to the original $O(m^{-(1+1/d)})$, where $m$ is the number of resampled particles. This improved rate is the lowest for ordered stratified resampling schemes, as conjectured in \citet{gerber2019negative}. We also present an almost sure bound on the Wasserstein distance between the original and Hilbert-curve-resampled empirical distributions. In light of these results, we show that, for $d>1$, the mean square error of sequential quasi-Monte Carlo with $n$ particles can be $O(n^{-1-4/[d(d+4)]})$ by implementing Hilbert curve resampling and selecting a specific low-discrepancy set. To the best of our knowledge, this is the first known convergence rate lower than $o(n^{-1})$.
\end{abstract}

{\small {\bf Keywords.}  Hilbert space-filling curve, particle filter, resampling, sequential Monte Carlo (SMC), stratification}


\section{Introduction}
\label{sec:intro}

{Sequential Monte Carlo dates back to the study of self-avoiding random walks \citep{hammersley1954poor,rosenbluth1955monte}, which is of great importance in chemistry and biology \citep{siepmann1992configurational,grassberger1997pruned}.} Sequential Monte Carlo has been studied intensively in the past two decades and applied broadly to  high-dimensional statistical inference, signal processing, biology and many other fields  \citep{liu1998sequential, doucet2001sequential}. Through building up the sampling (trial) distribution sequentially, a set of weighted samples can be used to approximate the  high-dimensional target distribution, or at least a certain aspect of it. 
The state-space model is a particularly interesting dynamic system that have been treated with sequential Monte Carlo. The  model is governed by the hidden Markovian state equation and the noisy observation equation. The hidden state, for instance, can represent the underlying volatility in an economical time series \citep{taylor2008modelling,gatheral2011volatility}, or the location in a terrain navigation problem \citep{bergman1999terrain,bergman2001posterior,gustafsson2002particle}, or many others.
In such models, characterizing the distribution of the hidden state is known as the filtering problem; and sequential Monte Carlo is more commonly known as the particle filter in this context \citep{gordon1993novel}.

Roughly speaking, sequential Monte Carlo is built based on sequential importance sampling, which recursively simulates a future state and reweighs the sampling path, with additional resampling steps \citep{liu1998sequential}. 
In a vanilla sequential importance sampling procedure, such as sequential imputation \citep{kong1994sequential}, weight degeneracy  arises as an inevitable problem. Since the importance weights are updated recursively at each step, stochastically most of the total weights  will concentrate on a very few samples, leading to exponentially increasing variance \citep{kong1994sequential}. 
One effective strategy to avoid weight degeneracy is to resample from the current samples according to the corresponding weights. 
Resampling alone does not provide any information for estimation at the current step, but only introduces additional randomness. 
The main intuition behind resampling is that particles with small weights are deemed less hopeful and thus discarded so as to save resources in order to explore  regions that may be more promising for the future \citep{liu1995blind}. 
Incidentally,  in the bootstrap filter of \cite{gordon1993novel}, every forward simulation step is followed immediately with a resampling step without investigating its advantages and disadvantages. \cite{liu1995blind} provided {an early} attempt at analyzing resampling (termed as rejuvenation in that article) {for statistical models}, providing some useful insights, but was short of a rigorous theory.

Each iteration of sequential Monte Carlo consists of two steps:  forward-sampling (or more intuitively, growth) and resampling. 
In the resampling step, we rejuvenate all the weights where samples with higher weights are more likely to be retained. 
In the growth step, we generate samples from the trial distribution and calculate the corresponding weight for each sample. Intuitively, the trial distribution should be as close to the target distribution as possible so as to explore  the relevant part of the sample space.

There are various means to resample from a collection of weighted particles. {Informally, one would like to minimize the ``resampling randomness'' over a certain class of valid resampling schemes. This goal is closely related to the balanced sampling design in survey sampling \citep[Chapter~8]{tille2006sampling}, which seeks to reduce the sampling variance using auxiliary variables.} The na\"ivest way to resample is called bootstrap resampling or multinomial resampling \citep{gordon1993novel}, where the new particles are sampled from independent and identically distributed (i.i.d.) multinomial distributions based on the original particle weights. Residual resampling \citep{liu1998sequential} and stratified resampling \citep{kitagawa1996monte} are two more popular resampling schemes in practice. {These methods have also been studied and used in scientific fields outside of statistics under different names of resampling (e.g., parent selection for genetic algorithms \citep[Chapter~4.2]{brindle1980genetic} and  stochastic reconfiguration in physics \citep[Chapter~10.3]{gubernatis2016quantum}).} \citet{douc2005comparison} compared the above resampling schemes and concluded that residual resampling and stratified resampling always have a smaller conditional variance than multinomial resampling does. For discrete state-spaces, the optimal resampling method \citep{fearnhead2003line} offers an interesting way of diversified sampling. 
Besides these traditional resampling schemes, \citet{reich2013nonparametric} proposed optimal transport resampling, an approach borrowing ideas from transportation theory. {However, there has been no theoretical guarantee for the optimal transport resampling (aside from its validity), to the best of our knowledge.} Recently, \citet{gerber2019negative} showed that stratified resampling after ordering the particles by the Hilbert space-filling curve has a relatively low conditional variance in some cases, which is also one of our interests in this article.

Sequential quasi-Monte Carlo introduced in \citet{gerber2015sequential} is a class of algorithms taking advantage of Hilbert curve resampling and quasi-Monte Carlo point sets. By constructing a low-discrepancy set on a product space, sequential quasi-Monte Carlo combines resampling and growth and numerically outperforms sequential Monte Carlo significantly. Theoretically, however, the convergence rate in terms of the mean squared error has only been shown to be $o(n^{-1})$ for certain low-discrepancy sets. It is naturally believed that the rate could be improved and should depend on the dimension $d$.

We focus on  theoretical properties of various resampling schemes and sequential quasi-Monte Carlo in this paper. 
We show that, in one dimensional cases, optimal transport resampling is equivalent to stratified resampling on the sorted particles, which minimizes the resampling variance as well as the expected squared energy distance between the empirical distributions before and after resampling. 
In $d$ dimensions, a natural generalization of ordered stratified sampling in one dimension is Hilbert curve resampling \citep{gerber2019negative}, which is stratified resampling on particles sorted using the Hilbert space-filling curve. We prove that its resampling variance is of the order $O(m^{-(1+2/d)})$ when $d>1$, where $m$ is the number of resampled particles. This improves the original rate $O(m^{-(1+1/d)})$. We show that the order cannot be further improved by resorting to a different ordering rule, confirming a conjecture in \citet{gerber2019negative}. We also derive a bound on the Wasserstein distance between the empirical distributions before and after Hilbert curve resampling. Based on the theoretical results on resampling, we further design a low-discrepancy set for sequential quasi-Monte Carlo and prove that the mean squared error under this set is of the order $O(n^{-1-[4/d(d+4)]})$ for $d>1$. This improves the original rate $o(n^{-1})$. We believe this low-discrepancy set captures some key intuitions of quasi-Monte Carlo and the tools can be modified to analyze other low-discrepancy sets as well.

The rest of the article is organized as follows. We provide some preliminaries, including relevant notations, definitions, and formulations, in Section~\ref{sec:pre}. In Section~\ref{sec:1d-optimal}, we prove the equivalence of several aforementioned resampling approaches in the one dimensional case. In Section~\ref{sec:multiple-dim}, we give upper bounds for the resampling error of Hilbert curve resampling in terms of both variance and Wasserstein distance. In Section~\ref{sec:multiple-des}, we focus on exploring  sequential quasi-Monte Carlo and derive a better convergence rate based on the theoretical results in Section~\ref{sec:multiple-dim}. We wrap up the paper in Section~\ref{sec:discussion} with some important open problems. All proofs are deferred to the supplement.



\section{Preliminaries}
\label{sec:pre}
\subsection{Notations}

We use superscript to denote the temporal notation (i.e., the step or iteration) and subscript for the sample index; the temporal notations are omitted for the sake of clarity whenever there is no confusion. The target distribution is denoted as $\pi(x)$, while $g(x)$ denotes the trial distribution in the sense of importance sampling, which is constructed in a forward sampling (growth) fashion in sequential Monte Carlo. When written without a subscript, $X$ and $W$ mean $(X_1,X_2,\dots,X_n)$ and $(W_1,W_2,\dots,W_n)$ for an appropriate $n$, and the set of tuples $(X_j, W_j)_{j=1}^n$ refers to a set of weighted samples, where $W_j \ge 0, j = 1,2,\cdots, n$.  Unless stated otherwise, the $W_j$'s are normalized so that $\sum_{j=1}^n W_j = 1$. We use $\tilde{X}_1, \tilde{X}_2, \cdots, \tilde{X}_m$ to denote the equally weighed samples after resampling, so that in some sense, $\sum_{i=1}^mm^{-1}\delta_{\tilde{X}_i}\approx\sum_{j=1}^n W_j\delta_{X_j}$,
where $\delta_x$ denotes the Dirac measure at point $x$. If $X_j\in\mathcal X$ for $j=1,2,\dots,n$, we use $\mathcal X^n$ to denote the space in which $X$ lives. We use $Z\sim\text{Multinomial}(1,y,p)$ to mean that $\p(Z=y_i)=p_i$, where $p$ is a probability vector. We write $m_d(\cdot)$ for the Lebesgue measure in $d$ dimensions. The standard $L_2$ norm is denoted as $\|\cdot\|$. For a vector $a$, $\diag(a)$ represents the diagonal matrix with the $i$th diagonal element being $a_i$. For a real number $u$, $\lfloor u \rfloor$ denotes the greatest integer less than or equal to $u$. The symbol $\simiid$ denotes sampling independent and identically distributed random variables.

\subsection{Sequential Monte Carlo}
To set up future analyses, we here describe a generic sequential Monte Carlo procedure. Let the target distribution $\pi(x)$ be supported in a $T$-dimensional space, which can be viewed as a joint distribution of a sequence of variables, say $\pi(x^{(1:T)})$. We can sample sequentially from a sequence of distributions $\{\pi_t(x^{(1:t)})\}_{t=1}^T$, where $\pi_T=\pi$. A generic sequential Monte Carlo algorithm is outlined in Algorithm~\ref{algorithm:SISR}.
\vspace*{-.3cm}

\begin{algorithm}
\caption{Sequential importance sampling with resampling.}
\label{algorithm:SISR}
\textbf{Input}: A sequence of target distributions $\{\pi_t(x^{(1:t)})\}_{t=1}^T$\\
\textbf{Output}: weighted particles $(X_{i}^{(1:T)},W_{i}^{(T)})_{1\le i\le n}$\\
    At time $t = 1$,\\
    \Indp Draw $X_{1}^{(1)}, \cdots, X_{n}^{(1)}$ from $g_1(X^{(1)})$.\\
    Calculate and normalize the importance weight: $W_{j}^{(1)} \propto {\pi_1(X_{j}^{(1)})}/{g_1(X_{j}^{(1)})}$.\\
    Resample $\tilde{X}_1^{(1)}, \tilde{X}_2^{(1)}, \cdots, \tilde{X}_n^{(1)}$ from $X_{1}^{(1)}, \cdots, X_{n}^{(1)}$ with probabilities $W^{(1)}_1, \cdots, W^{(1)}_n$, and reweight the samples $\tilde{X}_1^{(1)}, \tilde{X}_2^{(1)}, \cdots, \tilde{X}_n^{(1)}$ equally with $1/n$.\\
    Let $X_j^{(1)}=\tilde X_j^{(1)}$ for $j=1,2,\dots,n$.\\
	\Indm \SetAlgoLined\DontPrintSemicolon
	{
		\For{$t=2$ \KwTo $T$} {
        Draw $X_{j}^{(t)}$ from $g_t(X^{(t)}\mid X_j^{(1:t-1)})$ for $j=1,2,\dots,n$ conditionally independently.\\
        Calculate and normalize the importance weight:\\
        $$W_{j}^{(t)} \propto \frac{\pi_t\left(X_{j}^{(1:t)}\right)}{\pi_{t-1}\left(X_j^{(1:t-1)}\right)g_t\left(X_{j}^{(t)}\mid {X_j^{(1:t-1)}}\right)}$$\\
        \If{$t<T$} {Resample $\tilde{X}_1^{(1:t)}, \tilde{X}_2^{(1:t)}, \cdots, \tilde{X}_n^{(1:t)}$ from $X_{1}^{(1:t)}, \cdots, X_{n}^{(1:t)}$ with probabilities $W_1^{(t)}, \cdots, W_n^{(t)}$, and reweight the samples $\tilde{X}_1^{(1:t)}, \tilde{X}_2^{(1:t)}, \cdots, \tilde{X}_n^{(1:t)}$ equally with $1/n$.\\
        Let $X_j^{(1:t)}=\tilde X_{j}^{(1:t)}$.
        }
	}
	}
	Return $(X_{i}^{(1:T)},W_{i}^{(T)})_{1\le i\le n}$
\end{algorithm}

In the special case of a state-space model, we have
\begin{equation}
\begin{aligned}
Y^{(t)}\mid \left(X^{(1:t)}=x^{(1:t)},Y^{(1:t-1)}\right) &\sim p_y(\cdot\mid x^{(t)}),\\
X^{(t)}\mid \left(X^{(1:t-1)}=x^{(1:t-1)},Y^{(1:t-1)}\right)&\sim p_x(\cdot\mid x^{(t-1)}),t =2,\cdots, T,
\label{equation:state-space}
\end{aligned}
\end{equation}
where $p_x$ and $p_y$ represent distributions as well as density functions, $X^{(1)}, \cdots, X^{(T)}$ are unobserved hidden states, and  $Y^{(1)}, \cdots, Y^{(T)}$ are the observed sequence of variables. The filtering problem focuses on the target distribution
$$\pi_T(x^{(1:T)}) \propto \prod_{t=1}^{T} \left[p_x(x^{(t)}\mid x^{(t-1)}) p_y(y^{(t)}\mid x^{(t)})\right].$$
While implementing Algorithm~\ref{algorithm:SISR} in such a state-space model, the trial distribution at each step can be naturally (or na\"ively) chosen as $g_t(x^{(t)}\mid x^{(t-1)}) = p_x(x^{(t)}\mid x^{(t-1)})$, and thus the corresponding importance weight can be updated as
$w^{(t)} \propto w^{(t-1)}p_y (y^{(t)} \mid x^{(t)})$. 

\subsection{Resampling matrix}
\label{sec:resampling-matrix}
Suppose we have weighted particles $(W_j,X_j)_{j=1}^n$ with weights summing to one. Without loss of generality, we assume that the $X_j$'s are distinct since we can always merge particles with identical values and add up their weights. Consider the family of resampling methods indexed by a matrix $P_{m\times n}$, where the new unweighted particles $(\tilde X_i)_{i=1}^m$ are sampled independently from
\begin{equation*}
\tilde X_i\mid X,W\sim\text{Multinomial}(1,X,(p_{i1},p_{i2},\dots,p_{in})),
\end{equation*}
and $P$ has non-negative entries with $\sum_{i=1}^mp_{ij}=mW_j$ and $\sum_{j=1}^np_{ij}=1$. Note that permutating $P$'s rows does not change the resampling scheme. It can be easily verified that such a resampling strategy is unbiased, which means that for any $\phi$ we have
\begin{equation*}
\e\left[\frac1m\sum_{i=1}^m\phi(\tilde{X}_i)\mid X,W \right]=\frac1m\sum_{i=1}^m\sum_{j=1}^mp_{ij}\phi(X_j)=\sum_{j=1}^nW_j\phi(X_j).
\end{equation*}
We use $\mathcal P_{m,W}$ to denote the set of all matrices of this form and the set of all corresponding resampling methods, with slight abuse of notation. We call this collection of resampling methods {matrix resampling methods}, which also appears in \cite{reich2013nonparametric} and \cite{webber2019unifying}. The use of resampling matrices appeared at least as early as in \citet{hu2008basic} and also in many other works (e.g., \citet{reich2013nonparametric,whiteley2016role,webber2019unifying}). Most available resampling methods, as listed below, fit into this framework.

In multinomial resampling, each $\tilde{X}_i$ is an independent and identically distributed sample from the multinomial distribution $\text{Multinomial}(1,X,W)$. This corresponds to $p_{ij}=W_j$ for $i=1,\dots,m$, $j=1,\dots,n$, as shown in Figure~\ref{figure:Ps}(a).
In stratified resampling, we let $U_i \sim \text{Unif}\left((i-1)/m, i/m\right]$, independently for $i = 1, \dots, m$, and let
    $\tilde{X}_i = X_j$ if $U_i \in \left(\sum_{k=1}^{j-1} W_k, \sum_{k=1}^{j} W_k\right].$
    See Figure~\ref{figure:strat-res-ill} for an illustration.
    Stratified resampling corresponds to a staircase matrix; see Figure~\ref{figure:Ps}(b) for an example and Definition~\ref{def:stair-mat} for a formal definition.
In residual resampling, we first make $\lfloor mW_j \rfloor$ copies of  $X_j$ for all $j = 1, \dots, n$; then, apply multinomial or stratified resampling (corresponding to  Figure~\ref{figure:Ps}(c) and (d), respectively) for drawing the rest $m - \sum_{j=1}^n \lfloor mW_j \rfloor$ particles with $\tilde{W}_j \propto mW_j - \lfloor mW_j \rfloor$.
\begin{figure}
        \centering
        \includegraphics[width=.9\linewidth]{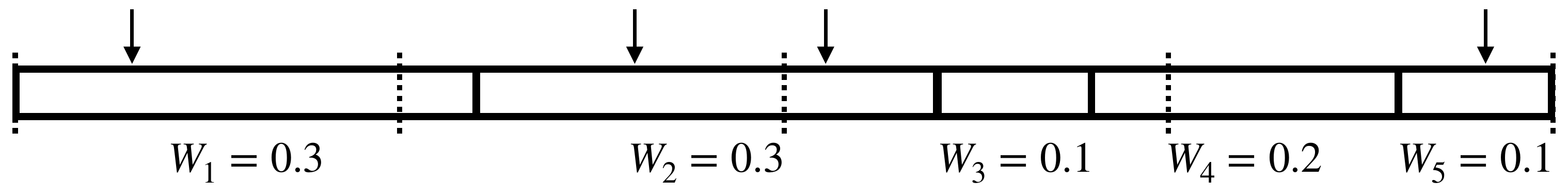}
        \caption{Illustration of stratified resampling. First line up the weights, then divide the interval into $m$ equal parts, uniformly choose one point from each subinterval and record in which weight's region it lands. In the presented example where $m=4$, $n=5$, particles $1$ and $5$ are resampled once, particle 2 is resampled twice and particles $3$ and $4$ are discarded.}
        \label{figure:strat-res-ill}
    \end{figure}

\begin{figure}
    \centering
    \begin{minipage}[!htbp]{.4\linewidth}
    $$\begingroup 
    \setlength\arraycolsep{4pt}
    \begin{pmatrix}
    0.3& 0.3 & 0.1& 0.2 &0.1\\
    0.3& 0.3 & 0.1& 0.2 &0.1\\
    0.3& 0.3 & 0.1& 0.2 &0.1\\
    0.3& 0.3 & 0.1& 0.2 &0.1
    \end{pmatrix}
    \endgroup$$
    \center{(a) Multinomial Resampling}
    \end{minipage}
    \begin{minipage}[!htbp]{.4\linewidth}
    $$\begingroup 
    \setlength\arraycolsep{4pt}
    \begin{pmatrix}
    1&  &  &  &\\
    0.2& 0.8 & & & \\
    & 0.4 & 0.4& 0.2 &\\
    & & & 0.6& 0.4
    \end{pmatrix}\endgroup
    $$
    \center{(b) Stratified Resampling}
    \end{minipage}
    
    \vspace{0.5cm}
    \begin{minipage}{.4\linewidth}
    $$\begingroup 
    \setlength\arraycolsep{4pt}
    \begin{pmatrix}
    1&  &  & & \\
     & 1 & &  &\\
    0.1& 0.1 & 0.2&  0.4& 0.2\\
    0.1& 0.1 & 0.2& 0.4&  0.2\\
    \end{pmatrix}
    \endgroup$$
    \center{(c) Multinomial Residual Resampling}
    \end{minipage}
    \begin{minipage}{.4\linewidth}
    $$\begingroup 
    \setlength\arraycolsep{4pt}
    \begin{pmatrix}
    1&  &  & & \\
     & 1 & &  &\\
    0.2& 0.2 & 0.4&  0.2& \\
    &  &  &  0.6&  0.4\\
    \end{pmatrix}
    \endgroup$$
    \center{(d) Stratified Residual Resampling}
    \end{minipage}
    \caption{Examples of resampling matrices with $m=4$ and $n = 5$, and particle weights $(W_1, W_2, W_3, W_4, W_5) = (0.3, 0.3 ,0.1, 0.2, 0.1)$.}
    \label{figure:Ps}
\end{figure}

\subsection{Criteria for choosing resampling schemes}
\label{sec:criteria}
To choose from the set of valid resampling procedures, we need a measure of goodness of a resampling procedure. Let $\p=\sum_{j=1}^nW_j\delta_{X_j}$ and $\tilde{\p}=\sum_{i=1}^mm^{-1}\delta_{\tilde{X}_i}$. It is natural to favor a stable process, where $\tilde{\p}$ is close to $\p$. Explicitly, we want to minimize $\e[\ell(\p,\tilde{\p})\mid X,W]$ for a loss function $\ell$. For example, we can pick $\ell(\p,\tilde{\p})$ to be $(\e_\p[\phi(X)]-\e_{\tilde{\p}}[\phi(X)])^2$ and use the conditional variance $\Var[m^{-1}\sum_{i=1}^m\phi(\tilde{X}_i)\mid X,W]$ as a measure of goodness. We can also choose $\ell$ to be the squared energy distance, which has the advantage of explicit expression and the property that the energy distance is zero if and only if two distributions are the same. The energy distance \citep{rizzo2016energy} between distributions $\p_1$ and $\p_2$ is defined as the square root of
\begin{equation*}
D^2(\p_1,\p_2)=2\e[\|Y_1-Y_2\|]-\e[\|Y_1-Y_1'\|]-\e[\|Y_2-Y_2'\|],
\end{equation*}
where $Y_1$ and $Y_1'$ follow $\p_1$, $Y_2$ and $Y_2'$ follow $\p_2$, and the four random variables are mutually independent. Another example is the Wasserstein distance, defined between distributions $\p_1$ and $\p_2$ as
\begin{equation*}
W_p(\p_1,\p_2)=\left(\inf_{\gamma\in\Gamma(\p_1,\p_2)}\e_{(Y_1,Y_2)\sim\gamma}[\|Y_1-Y_2\|^p]\right)^{1/p}, \ \  p\ge1,
\end{equation*}
where $\Gamma(\p_1,\p_2)$ denotes all probability measures that have $\p_1$ and $\p_2$ as their marginal distributions.

In Section~\ref{sec:1d-optimal}, we prove that minimizing the conditional variance is equivalent to minimizing the expected squared energy distance  in  one dimensional cases, both of which can be achieved by ordered stratified resampling (i.e., stratified resampling on the sorted particles). In Section~\ref{sec:multiple-dim}, we give upper bounds for conditional variance and expected Wasserstein distance for ordered stratified resampling, where the particles are sorted according to the Hilbert curve in multiple dimensions.

\section{Optimal resampling in one dimension}
\label{sec:1d-optimal}
A good resampling scheme should ideally incorporate the information of the state values $X_j$'s, since the loss function usually depends on them. 
In this section, we show that, by incorporating the $X_j$'s value information, the stratified resampling method minimizes several objectives proposed in the literature. Note that in this section, we consider the case where the particles take values in a one dimensional space. For example, resampling in a state-space model where the hidden state at each step is one-dimensional. In this case, we can focus on the last dimension of each particle, since the other components will not affect the future.

To study the stratified resampling matrix, we first define the staircase matrix{, which is the same as a stratified resampling matrix as we show in Proposition~\ref{prop: staircase}}. This will help with understanding why ordering the states before applying stratified resampling can lower the resampling variance.
\begin{definition}[Staircase matrix]
\label{def:stair-mat}
We call a matrix $P$ \emph{staircase matrix} if the following conditions are satisfied:
\begin{enumerate}
    \item[(1)] In each row and column of $P$, non-zero entries are consecutive. In other words, if $p_{ij_1} \neq 0$ and $p_{ij_2} \neq 0$ for $j_1 < j_2$, then for all $j_1 < j < j_2$, $p_{ij}\neq 0$, and similarly for the columns.
    \item[(2)] 
    For any quadruplet $(i,j,k,l)$ such that $i < k, j<l$, at least one of $p_{il}$ and $p_{kj}$ is $0$.
\end{enumerate}
\end{definition}

{ 
\begin{proposition}\label{prop: staircase}
Any stratified resampling scheme corresponds to a unique staircase matrix up to row permutations.
\end{proposition}
}
A staircase matrix has at most $n+m-1$ non-negative entries and has a clear spatial structure.
The non-negative entries form a path (allowing diagonal moves) from the top left entry to the bottom right entry.

\begin{lemma}\label{lemma:unique_of_sm}
For $m,n>2$, there can only be one unique $m$ by $n$ staircase matrix that has non-negative entries and satisfies:
$$\sum_{j=1}^n p_{ij} = r_i>0 \text{ and } \sum_{i=1}^m p_{ij} = c_j>0$$
\end{lemma}

By Lemma~\ref{lemma:unique_of_sm}, the staircase resampling matrix is unique given the weights for each particle. Then we can define a stratified resampling matrix.

\begin{definition}[Stratified resampling matrix]\label{def:ordered-sr}
We call a matrix $P^\text{SR}_{m,W}\in\mathcal P_{m,W}$ the {stratified resampling matrix} of a set of weighted particles $(X_j, W_j)_{j=1}^n$ if 
$P^\text{SR}_{m,W}$ can be converted to a staircase matrix after some row permutation.
\end{definition}

\begin{theorem}
\label{theorem: 1d-resampling}
For particles $(X_j,W_j)_{j=1}^n$ with $X_1 < X_2 < \cdots X_n$, resampling defined by $P^\text{SR}_{m,W}$ minimizes the following objectives:
\begin{enumerate}
    \item[(i)] The conditional variance $\Var_P\left[{\frac1m}\sum_{i=1}^m \tilde{X}_i\mid X,W\right]$.
    \item[(ii)] The expected squared energy distance $E_P\left[D^2\left(\sum_{i=1}^m m^{-1}\delta_{\tilde X_i}, \sum_{j=1}^n W_j \delta_{X_j}\right)\right]$.
    \item[(iii)] The earth mover distance $\sum_{i=1}^m \sum_{j=1}^n p_{ij} \ell(Y_i - X_j)$ where $ell$ is a strictly convex function, and $Y_1 < \cdots < Y_m$ is any given sequence of ascending numbers.
\end{enumerate}
\end{theorem}

\begin{remark}\label{remark: webber}
If the goal is to estimate $\e[\phi(X)]$, then ordering the states by function $\phi$ and then applying stratified resampling gives the minimum variance. {This result also appeared in \citet{webber2019unifying}, where it was proved using an optimization argument. Our proof uses a similar idea and directly shows that when the resampling variance is minimized, the resampling matrix must be a staircase matrix and corresponds to ordered stratified resampling. A similar approach is used to prove (iii) as well.}
\end{remark}

\section{Error of ordered stratified resampling}
\label{sec:multiple-dim}

In this section, we analyze the error induced by ordered stratified resampling. 
\begin{theorem}\label{theorem:variance_1d}
Suppose one-dimensional particles $(\tilde X_{i})_{i=1}^m$ is resampled with ordered stratified resampling from $(X_j,W_j)_{j=1}^n$, then for any Lipschitz function $\phi$ with coefficient $L_\phi$,
\begin{equation*}
\Var \left[\dfrac{1}{m}\sum_{i=1}^m \phi(\tilde X_i)\mid{X}, W\right] 
\leq \dfrac{L_\phi^2}{4m^2} (\max_{1\le i\le n} {X_i} - \min_{1\le i\le n} {X_i})^2.
\end{equation*}
\end{theorem}


We here provide some intuition behind ordered stratified sampling. Since the new particles are sampled independently, we only need to make sure that each new particle brings in little randomness. It is easy to see from the staircase structure of the resampling matrix that each $\tilde X_i$ takes value in a sequence of consecutive $X_j$'s. Since the original particles have been ordered, this sequence of $X_j$'s are close to each other in the space. Together with the fact that $\phi$ is Lipschitz, we see that for each $i$, $\phi(\tilde X_i)$ is bounded in a small region. 

In multiple dimensions, it has been noticed that the Hilbert space-filling curve \citep{hilbert1935stetige} can help lower the sampling variance \citep{gerber2015sequential,he2016extensible,gerber2019negative}. In particular, \citet{gerber2019negative} used the Hilbert curve in the context of resampling. They showed that the resampling variance for Lipschitz functions with $m$ particles is of order $O(m^{-(1+1/d)})$, where $d$ is the number of dimensions. We improve this bound to $O(m^{-(1+2/d)})$ and show that this new rate is the best for ordered stratified resampling schemes with any ordering, as conjectured in \citet{gerber2019negative}.

A $d$-dimensional Hilbert curve is a continuous function $H:[0,1]\to[0,1]^d$. Its most important properties relevant to our tasks are as follows:
\begin{enumerate}
    \item[(i)] $H$ is surjective.
    \item[(ii)] $H$ is H\"older continuous with exponent $1/d$ \citep{he2016extensible}:
        \begin{equation*}
        \|H(x)- H(y)\| \leq 2\sqrt{d+3} |x- y|^{1/d}.
        \end{equation*}
    \item[(iii)] $H$ is measure-preserving. For each Lebesgue measurable $I\subseteq[0,1]$, $m_1(I)=m_d(H(I))$.
\end{enumerate}
The Hilbert curve is defined as the limit of a sequence of curves; see Figure~\ref{figure:Hilbert-curve} for an illustration in two and three dimensions. Many software packages can efficiently convert between $x$ and $H(x)$ (e.g., the Python package {\small\textsf{hilbertcurve}}). {In practice, the computation cost of this approximation is quite minimal compared to the sampling part.} We omit here the rigorous definition of Hilbert curves and refer  interested readers to \citet{sagan2012space}. For the purpose of resampling, the most relevant property is the H\"older continuity. This ensures that $H(I)$, the image of an interval $I\subseteq[0,1]$, has its diameter bounded above by $2\sqrt{d+3}\cdot m_1(I)^{1/d}$. As an illustration, we plot the images of $H([i/k,(i+1)/k])$ for $i=0,1,\dots,k-1$ and $k=5,6,7,8$ in Figure~\ref{figure:Hilbert-curve-multiple-parts}.

\begin{figure}[h]
\begin{subfigure}{.24\textwidth}
  \centering
  \includegraphics[width=.9\linewidth]{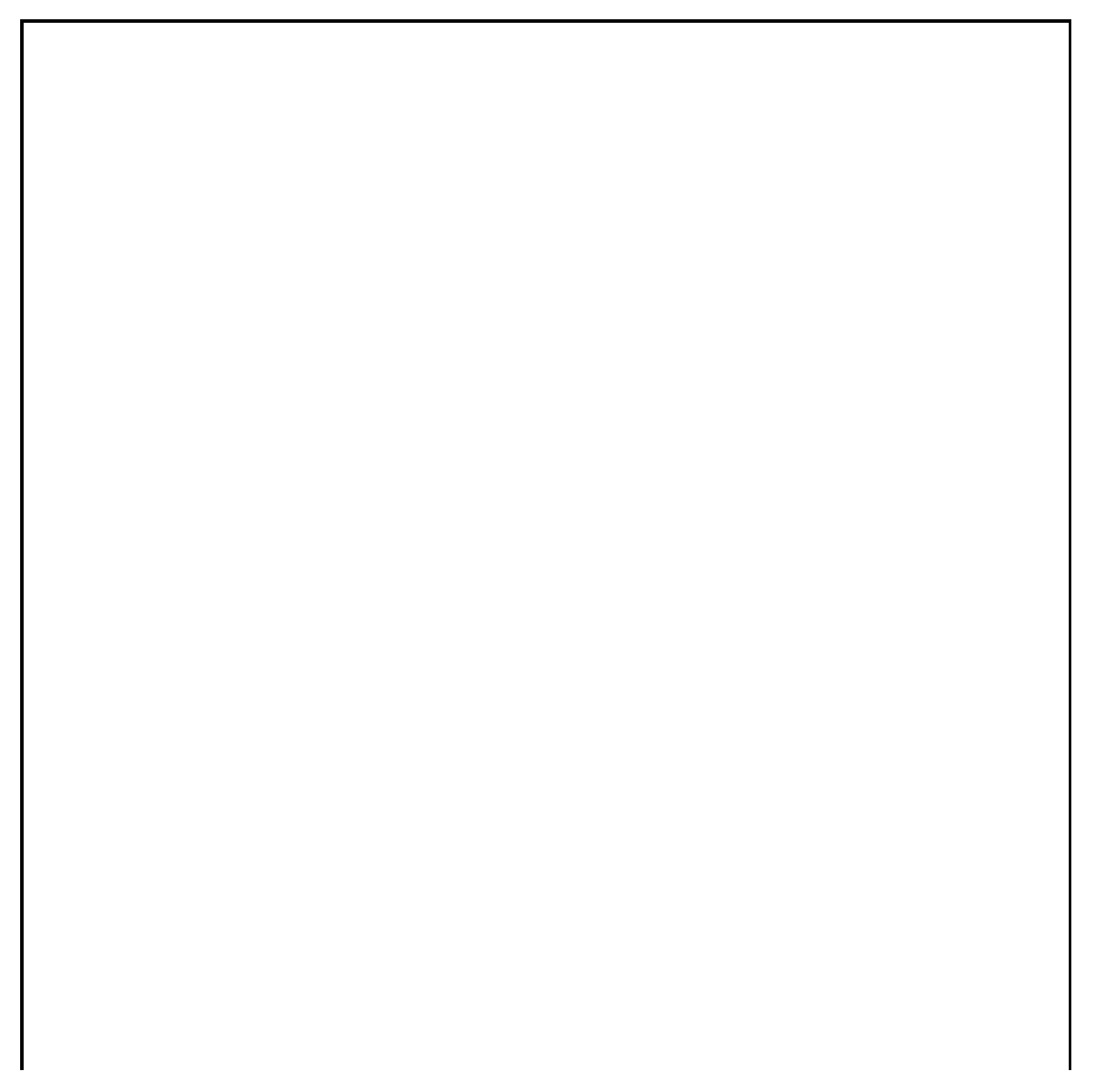}  
  \caption{$H_{2,1}$}
\end{subfigure}
\begin{subfigure}{.24\textwidth}
  \centering
  \includegraphics[width=.9\linewidth]{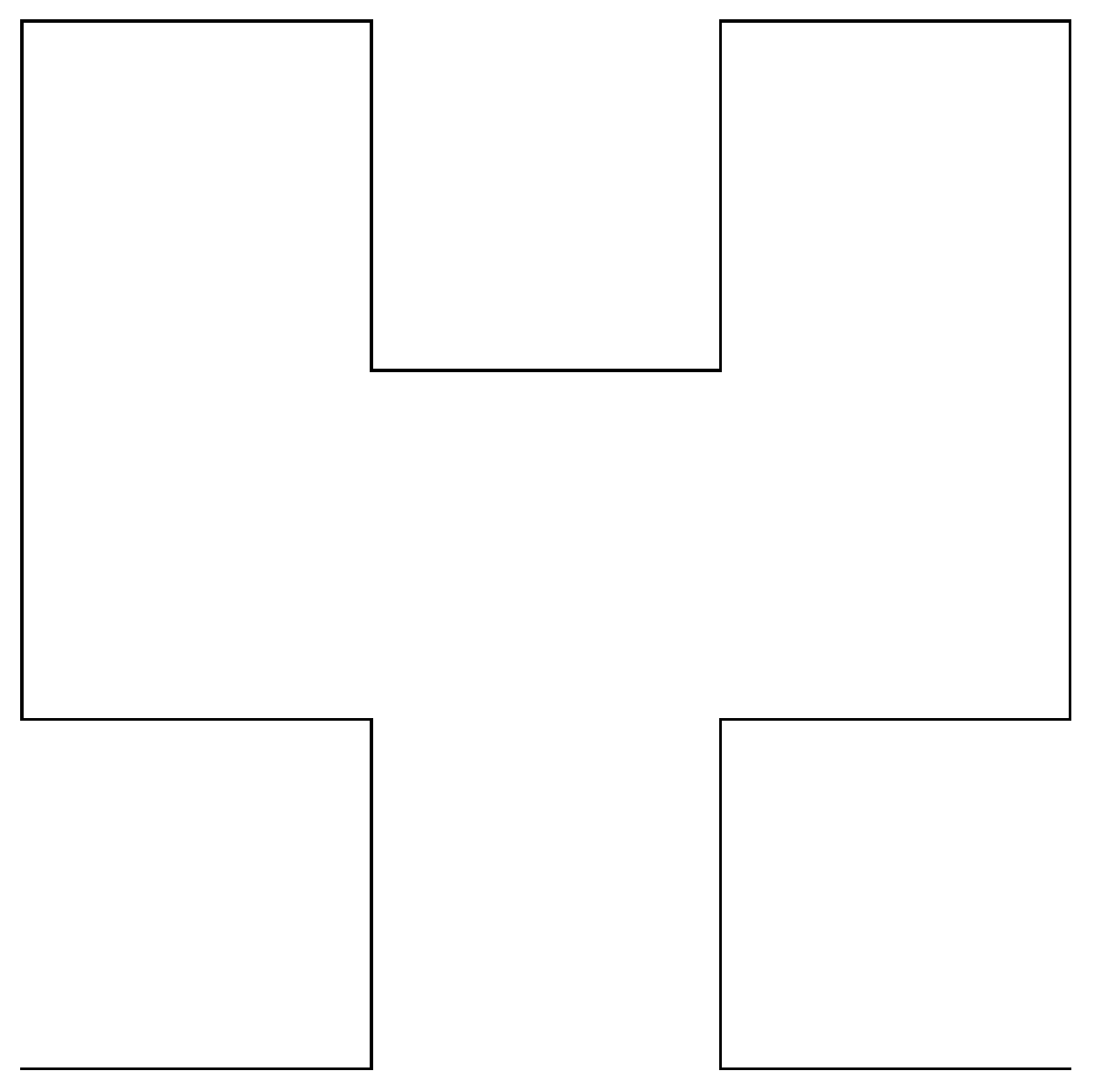}  
  \caption{$H_{2,2}$}
\end{subfigure}
\begin{subfigure}{.24\textwidth}
  \centering
  \includegraphics[width=.9\linewidth]{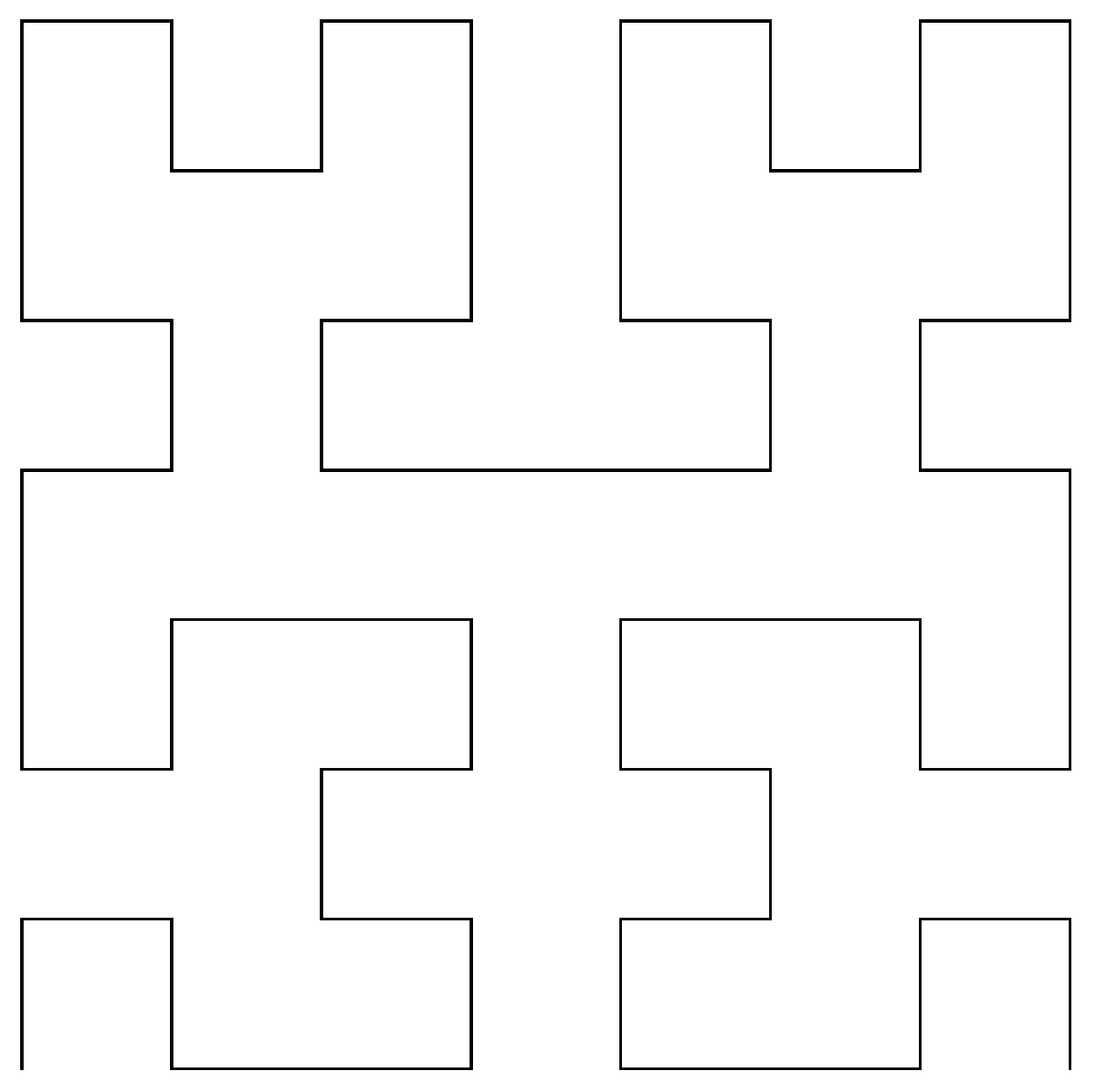}  
  \caption{$H_{2,3}$}
\end{subfigure}
\begin{subfigure}{.24\textwidth}
  \centering
  \includegraphics[width=.9\linewidth]{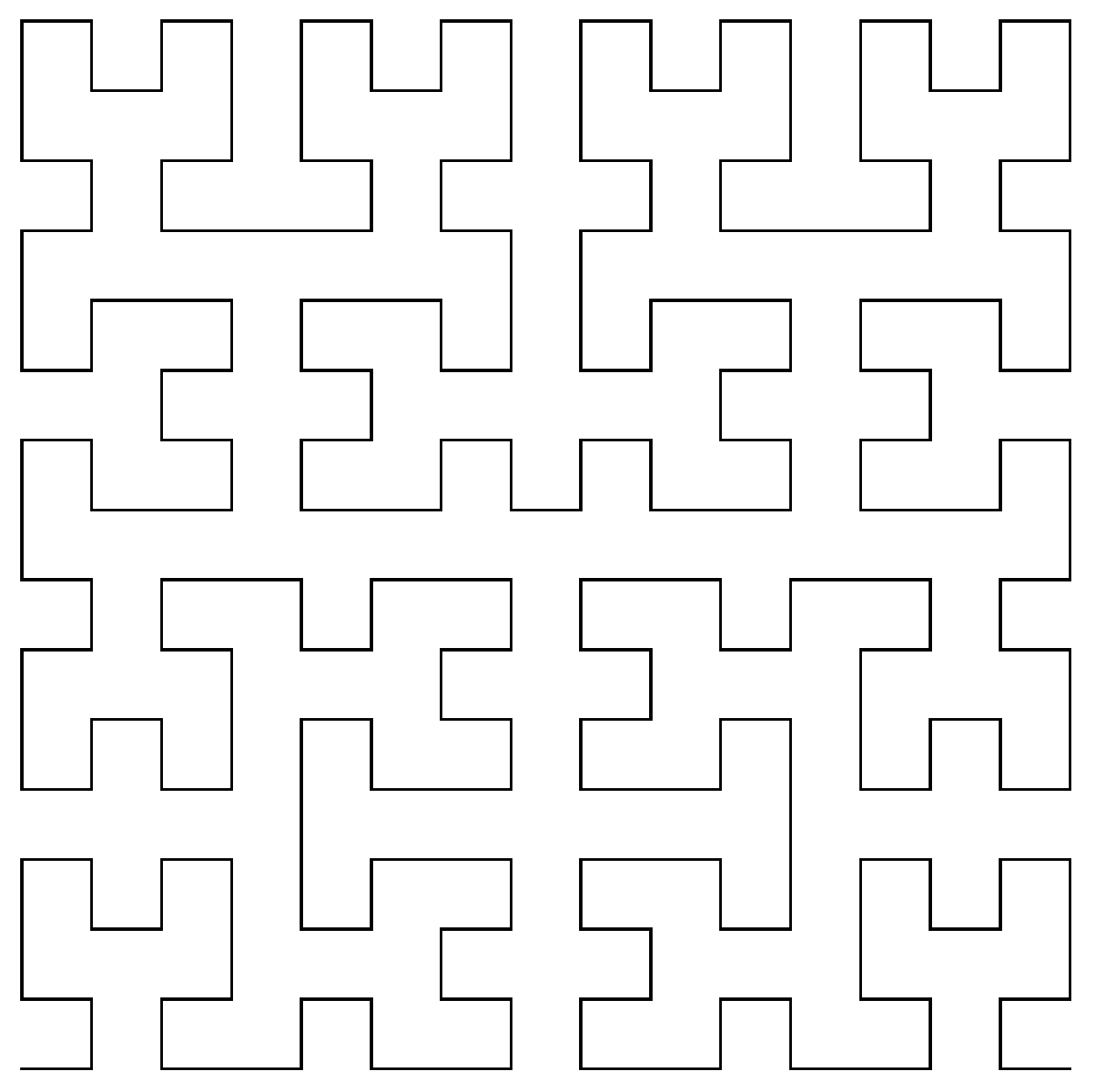}  
  \caption{$H_{2,4}$}
\end{subfigure}
\begin{subfigure}{.24\textwidth}
  \centering
  \includegraphics[width=.9\linewidth]{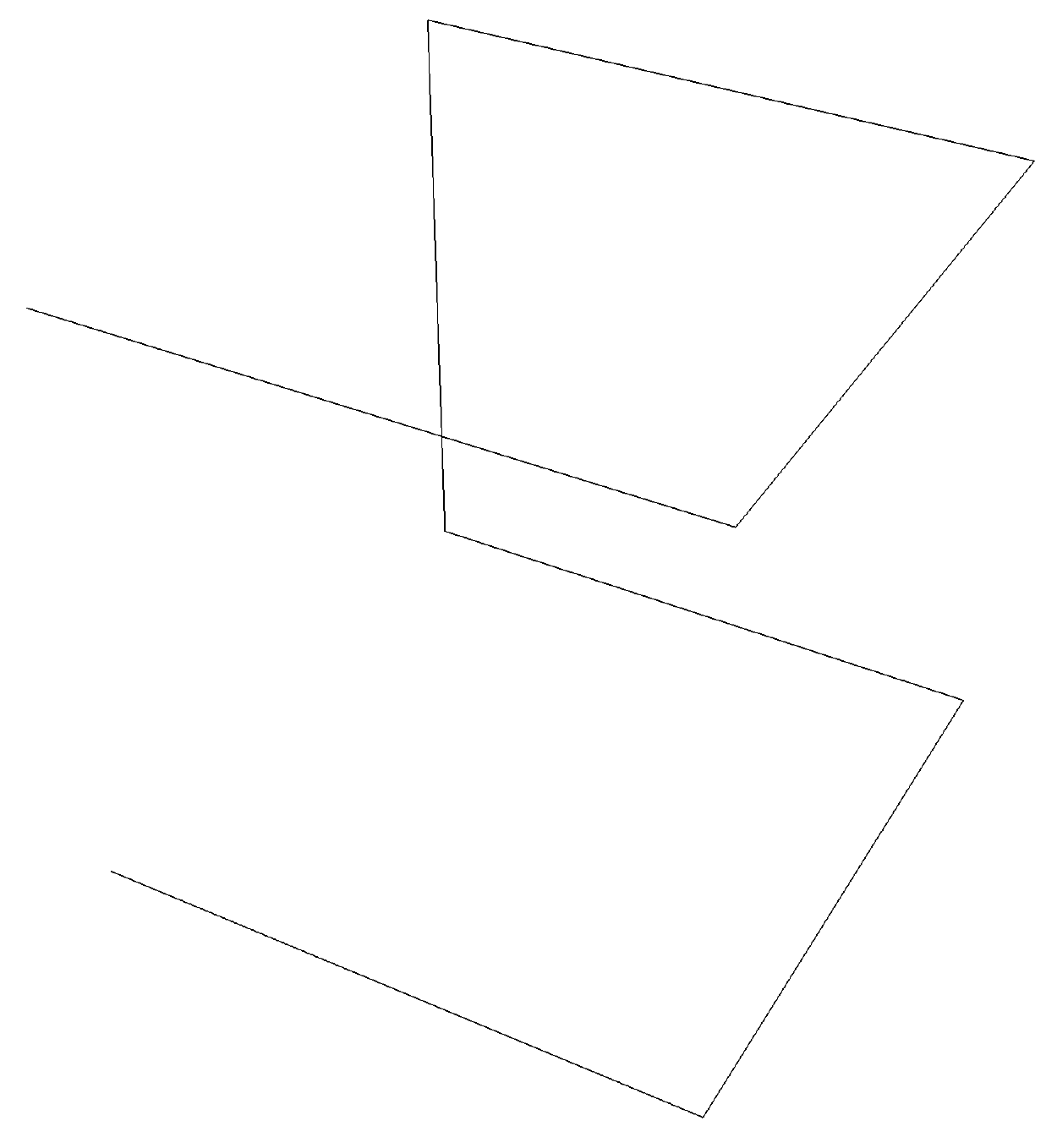}  
  \caption{$H_{3,1}$}
\end{subfigure}
\begin{subfigure}{.24\textwidth}
  \centering
  \includegraphics[width=.9\linewidth]{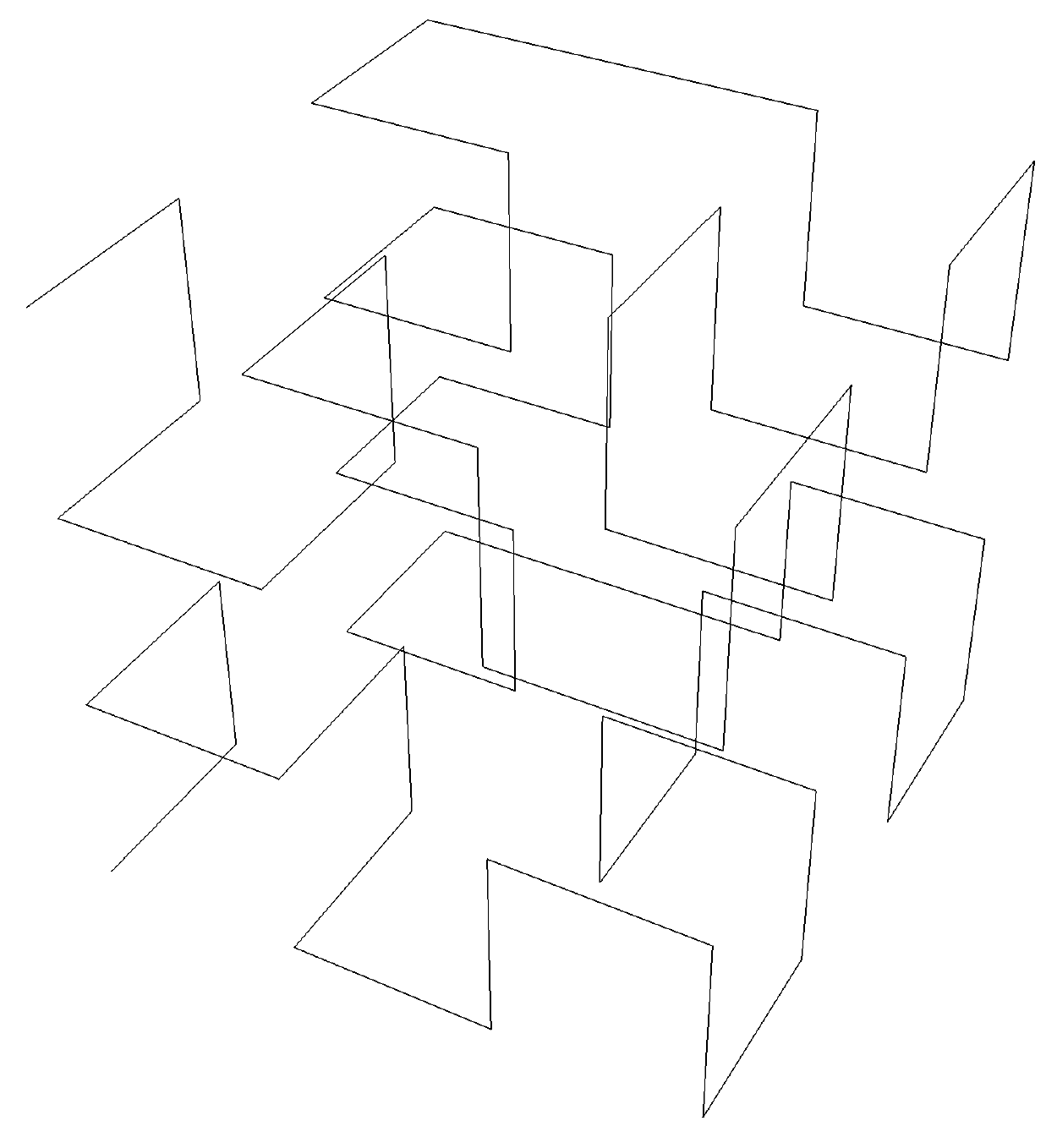}  
  \caption{$H_{3,2}$}
\end{subfigure}
\begin{subfigure}{.24\textwidth}
  \centering
  \includegraphics[width=.9\linewidth]{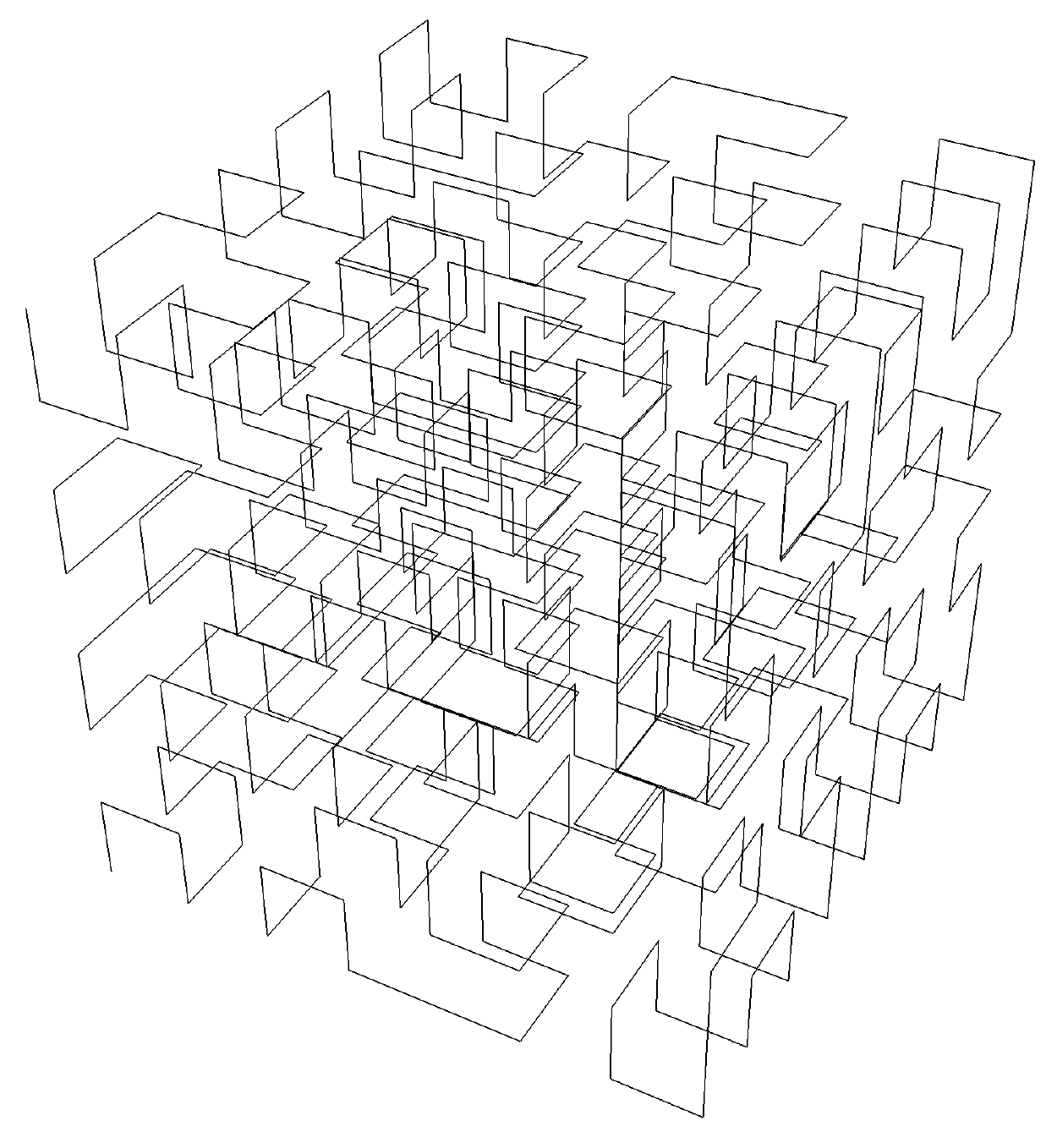}  
  \caption{$H_{3,3}$}
\end{subfigure}
\begin{subfigure}{.24\textwidth}
  \centering
  \includegraphics[width=.9\linewidth]{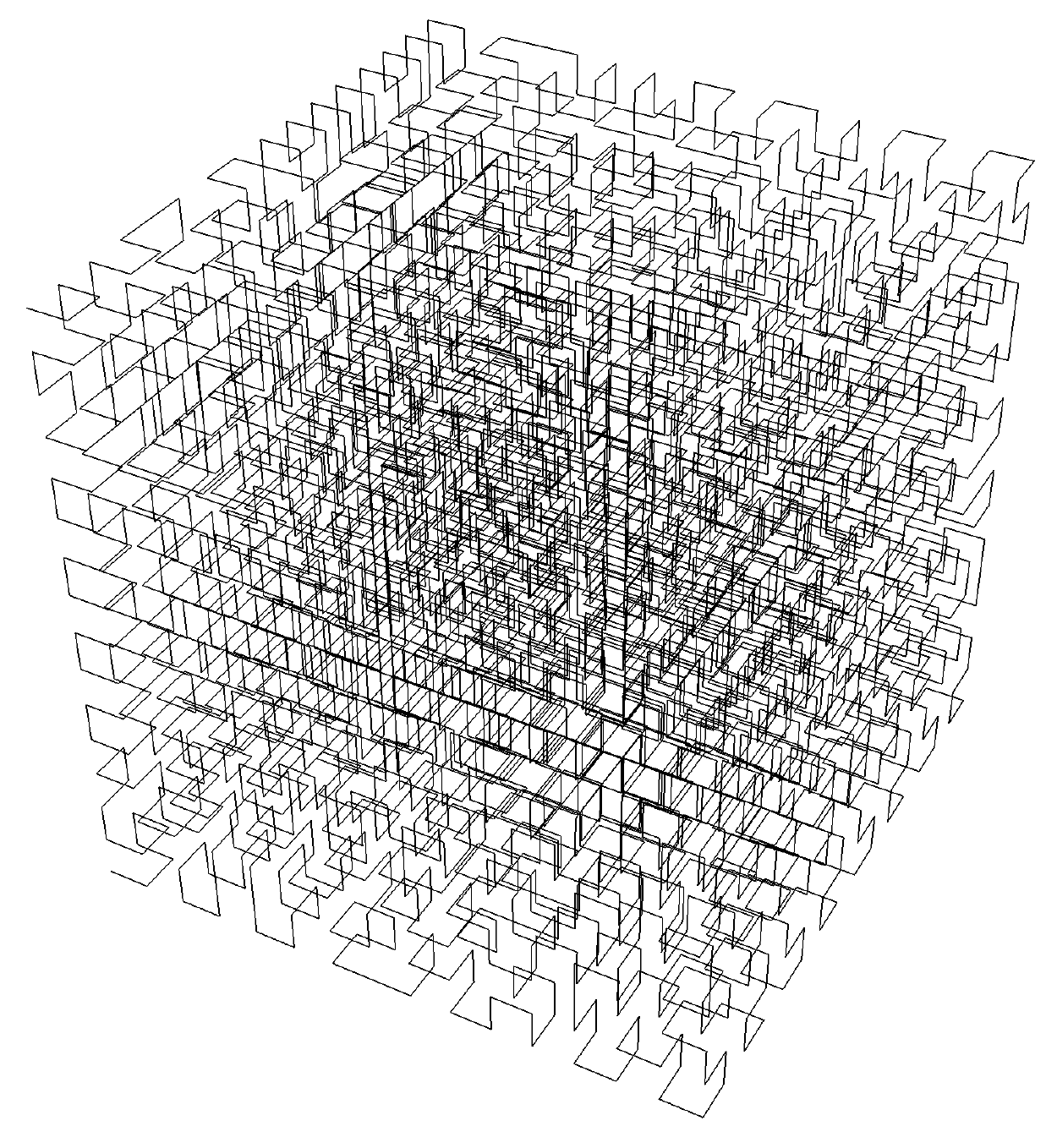}  
  \caption{$H_{3,4}$}
\end{subfigure}
\caption{Hilbert curves of the first four orders in two and three dimensions.}
\label{figure:Hilbert-curve}
\end{figure}

\begin{figure}[h]
\begin{subfigure}{.24\textwidth}
  \centering
  \includegraphics[width=.9\linewidth]{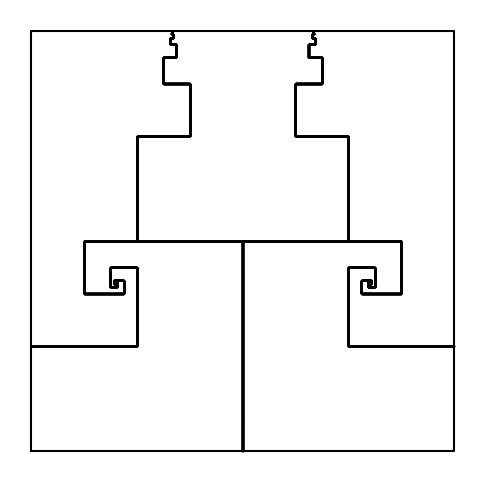}  
  \caption{Five parts.}
\end{subfigure}
\begin{subfigure}{.24\textwidth}
  \centering
  \includegraphics[width=.9\linewidth]{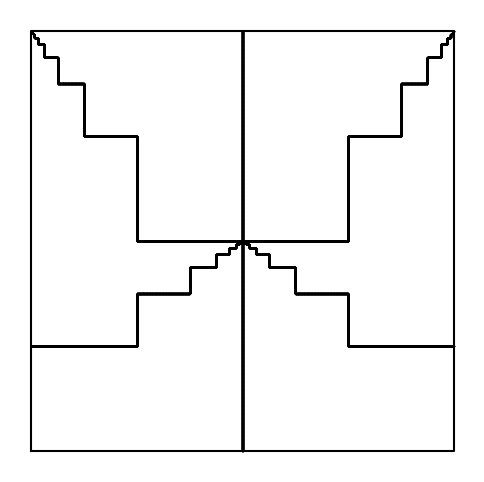}
  \caption{Six parts.}
\end{subfigure}
\begin{subfigure}{.24\textwidth}
  \centering
  \includegraphics[width=.9\linewidth]{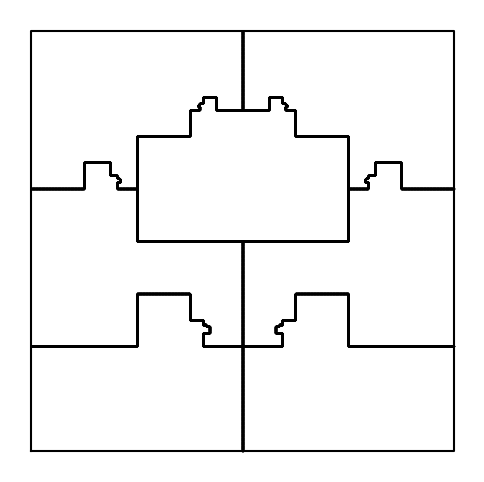}
  \caption{Seven parts.}
\end{subfigure}
\begin{subfigure}{.24\textwidth}
  \centering
  \includegraphics[width=.9\linewidth]{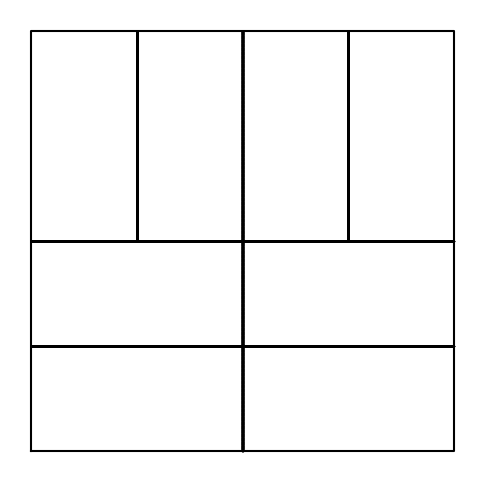}
  \caption{Eight parts.}
\end{subfigure}
\caption{The unit square divided into several parts with equal areas based on the Hilbert curve.}
\label{figure:Hilbert-curve-multiple-parts}
\end{figure}

Now we formally introduce the Hilbert curve resampling first proposed in \citet{gerber2019negative}. Proposition~2 in \citet{gerber2019negative} says that there exists a one-to-one Borel measurable function $h:[0,1]^d\to[0,1]$ such that $H(h(x))=x$ for all $x\in[0,1]^d$. The resampling procedure is to first sort the particles so that $(h(X_j))_{j=1}^n$ is in ascending order, and then apply stratified resampling. Note that in one dimension this reduces to ordered stratified sampling. Following the intuition in the one-dimensional case, each new particle is bounded in a small region in $[0,1]^d$ due to the H\"older continuity of $H$, which limits the variability of $\tilde X_i$. See Figure~\ref{figure:Hilbert-curve-multiple-parts-weights} for an illustration.
Theorem~\ref{theorem:variance_multi_dim} gives an upper bound on the resampling variance, which is an improved bound compared to the one reported in Theorem~5 in \citet{gerber2019negative}.

\begin{figure}[h]
\begin{subfigure}{.45\textwidth}
  \centering
  \includegraphics[width=.9\linewidth]{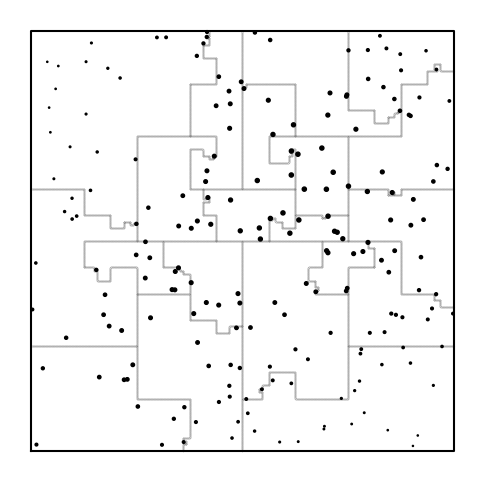}  
  \caption{$n=200$ particles resampled into $m=20$.}
\end{subfigure}
\begin{subfigure}{.45\textwidth}
  \centering
  \includegraphics[width=.9\linewidth]{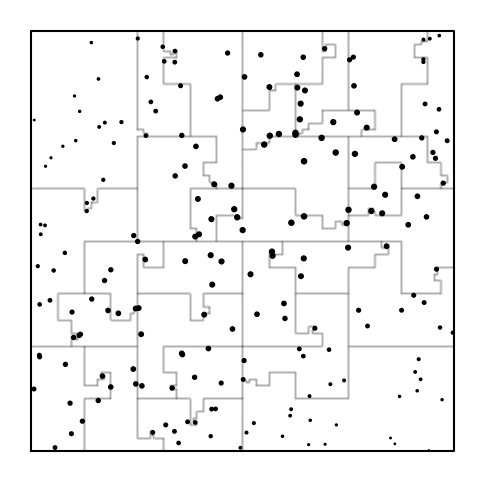}
  \caption{$n=200$ particles resampled into $m=30$.}
\end{subfigure}
\caption{The unit square divided into $m$ parts based on the Hilbert curve and the particle weights. Size of the point represents their particle weight. Each region contains particles with weights summing to one (neighbouring regions divide weights of the particles on the boundary).}
\label{figure:Hilbert-curve-multiple-parts-weights}
\end{figure}

\begin{theorem}\label{theorem:variance_multi_dim}
Let $\phi: [0,1]^d \rightarrow [0,1]$, $d>1$, be a Lipschitz function with Lipschitz coefficient $L_\phi$. If $(X_j)_{j=1}^n$ is sorted in an ascending order by the value of $h(X_j)$, then stratified sampling satisfies
$$\Var_{\textnormal{HC-strat}} \left[\dfrac{1}{m}\sum_{i=1}^m \phi(\tilde X_i)\mid{X},W\right] \le\dfrac{(d+3)L_\phi^2}{m^{1+2/d}}.$$
\end{theorem}

{ \begin{remark} \label{remark: common_structure}
The intuition behind Theorems~\ref{theorem:variance_1d} and \ref{theorem:variance_multi_dim} is the same: in stratified resampling, the variance of each individual resampled particle is controlled because it is sampled from a set of particles spatially close to each other. In fact, one can easily generalize Theorem~\ref{theorem:variance_multi_dim} to the H\"older function case: if $|\phi(x)-\phi(y)|\le L_\phi\|x-y\|^{\beta}$, $\beta\in(0,1]$, then $$\Var_{\textnormal{HC-strat}} \left[\dfrac{1}{m}\sum_{i=1}^m \phi(\tilde X_i)\mid{X},W\right] \le\dfrac{(d+3)L_\phi^2}{m^{1+2\beta/d}}.$$
\end{remark}}

\begin{remark}
The exponent $1+2/d$ in the theorem improves the original rate $1+1/d$ shown in \citet{gerber2019negative}.
It is conjectured in \citet{gerber2019negative} that the Hilbert curve is the best choice for ordering the particles. For clarity, we take the Lipschitz coefficient to be $1$ and $m=n$. Define the space of valid probability vector as 
\begin{equation*}
\Delta_n=\left\{(w_1,w_2,\dots,w_n)\in\rr^n:\sum_{j=1}^nw_j=1, w_i\ge0\text{ for all }1\le i\le n\right\}.
\end{equation*} Theorem~\ref{theorem:variance_multi_dim} implies that
\begin{equation*}
\limsup_{n\to\infty} n^{1+{\frac2d}}\sup_{X\in[0,1]^{d\times n}}\sup_{W\in\Delta_n}\sup_{\phi\in\Phi_d}\Var_{\textnormal{HC-strat}}\left[\frac1n\sum_{i=1}^n\phi(\tilde X_i)\mid {X}, {W}\right]\le d+3,
\end{equation*}
{where $\Phi_d$ denotes the set of $1$-Lipschitz functions from $[0,1]^d$ to $[0,1]$, $d>1$. For other space-filling curves (which may be cheaper to implement) with a different H\"older exponent, similar results hold with an exponent different from $1+2/d$. However, we show in Proposition~\ref{proposition:optimality} that no other ordering rule can improve the exponent $1+2/d$.}
\begin{proposition}
Let $\Phi_d$ be the set of $1$-Lipschitz functions from $[0,1]^d$ to $[0,1]$, $d>1$. Let $o(x):[0,1]^d\to[0,1]$ be a one-to-one function. The stratified sampling procedure after ordering particles by $o$ satisfies
\begin{equation*}
\limsup_{n\to\infty} n^{1+{\frac2d}}\sup_{X\in[0,1]^{d\times n}}\sup_{W\in\Delta_n}\sup_{\phi\in\Phi_d}\Var_{o\textnormal{-strat}}\left[\frac1n\sum_{i=1}^n\phi(\tilde X_i)\mid {X}, {W}\right]\ge\frac1{27d}.
\end{equation*}
\label{proposition:optimality}
\end{proposition}
\end{remark}

 Hilbert resampling is also stable in terms of the Wasserstein distance, as stated in Theorem~\ref{theorem:hc-wasserstein}. The Wasserstein distance is arguably a more intuitive notion to measure the stability of a resampling algorithm than conditional variance. When $p\le d$, Theorem~\ref{theorem:hc-wasserstein} is intuitively optimal, since $m$ balls with radius of the order $1/m^{1/d}$ are needed to cover the space.
\begin{theorem}
Under $d$-dimensional Hilbert curve resampling, $d\ge1$, the Wasserstein distance $W_p$ between $\tilde\p=\sum_{i=1}^mm^{-1}\delta_{\tilde X_i}$ and $\p=\sum_{j=1}^nW_j\delta_{X_j}$ is almost surely upper bounded by $2 \sqrt{d+3} m^{-{ \frac1{\max(p,d)} }} $.
\label{theorem:hc-wasserstein}
\end{theorem}

{
It is worthwhile to point out that in practice, depending on the targeted quantities of interest,  there might exist an effective dimension lower than $d$. For example, if we only care about functions of the first $\tilde d$ coordinates, we should sort the particles using the $\tilde d$-dimensional Hilbert curve and the first $\tilde d$ coordinates of the particles; if the particles concentrate on a  $\tilde d$-dimensional subspace, we should  project the particles to this subspace and sort the particles using the $\tilde d$-dimensional Hilbert curve.
}




\section{Mean square error of sequential quasi-Monte Carlo}
\label{sec:multiple-des}
\subsection{Sequential quasi-Monte Carlo}
\label{sec:sqmc}
In this section, we discuss how to utilize the previous results to obtain a new convergence rate for the sequential quasi-Monte Carlo proposed in \citet{gerber2015sequential}, which can be structured identically as Algorithm~\ref{algorithm:SISR} with the same weight computation, but with different resampling and growth steps.

Suppose there exists function $\Gamma_1(\cdot)$ and $\Gamma_t(\cdot,\cdot)$ for $2\le t\le T$ such that $\Gamma_1(V)\sim g_1(\cdot)$ and $\Gamma_t(X,V)\mid X\sim g_t(\cdot\mid X)$, where $V\sim\Unif([0,1]^d)$ is independent of $X$. Assume at the beginning of step $t$, we have weighted samples $(X_j^{(1:t-1)},W_j^{(t-1)})_{j=1}^n$, which have been ordered by the Hilbert mapping $h$ so that $h(X_1^{(t-1)})\le\cdots\le h(X_n^{(t-1)})$. Recall that Hilbert-curve stratified sampling can then be implemented by independently sampling $U_i\sim\Unif([(i-1)/n,i/n])$ for $1\le i\le n$ and let $\tilde X_i^{(t-1)}=X_{\sigma(U_i,W)}^{(t-1)}$, where $\sigma(U_i,W)=j$ if $\sum_{k=1}^{j-1}W_k<U_i\le\sum_{k=1}^jW_k$. Suppose we have a low-discrepancy set $U^{(t)}=\{(u_i,v_i):u_i\in[0,1],v_i\in[0,1]^d,1\le i\le n\}$, labeled in the way that $u_{1:n}$ are in ascending order. Intuitively speaking, a low-discrepancy set is a set that spreads evenly in $[0,1]^{1+d}$; see \citet{gerber2015sequential} for a more detailed discussion. Sequential quasi-Monte Carlo combines resampling and growth by defining
\begin{equation*}
X_j^{(t)}=\left\{
\begin{aligned}
&\Gamma_1(v_j), &t=1,\\
&\Gamma_t(X^{(t-1)}_{\sigma(u_j,W_{1:n}^{(t-1)})},v_j), &2\le t\le T,
\end{aligned}
\right.,\quad 1\le j\le n.
\end{equation*}
If the set $U^{(t)}$ contains $n$ independent samples from $\Unif([0,1]^{1+d})$, then we recover Algorithm~\ref{algorithm:SISR} with Hilbert resampling. It was shown in \citet{gerber2015sequential} that some choice of $U^{(t)}$ (e.g., the nested scrambled Sobol sequence) can achieve a mean square error of  order $o(n^{-1})$. Next, we will show that a specifically chosen set can achieve $O(n^{-1-4/[d(d+4)]})$.

\subsection{Stratified multiple-descendant growth}
\label{sec:SMG-theory}
The intuition behind Sequential quasi-Monte Carlo is that the consecutive resampled particles $(X^{(t)}_{\sigma(u_j,W_{1:n}^{(t)})})_{j=a}^b$ are close in space due to the H\"older continuity of the Hilbert curve, so if $v_{a:b}$ are more spread out, the space can be probed more consistently by stratified growth. The main difficulty of quantifying  the convergence rate of Sequential quasi-Monte Carlo lies at the deterministic or semi-deterministic nature of the set $U^{(t)}$. We exploit this intuition and construct a specific set that enables more careful convergence analysis.

Let $n=sr$, and let $U_k\sim\Unif[(k-1)/s,k/s]$ be independent for $1\le k\le s$. Let $V_{(k-1)s+\ell}=H(\tilde V_{k\ell})$, where $H$ is the $d$-dimensional Hilbert curve and $\tilde V_{k\ell}\sim\Unif[(\ell-1)/r,\ell/r]$, independently for $1\le k\le s$, $1\le\ell\le r$. We define $U_\textnormal{SMG}^{(t)}=\{(U_{\lfloor i/r\rfloor+1},V_i):1\le i\le n\}$. Here, SMG stands for stratified multiple-descendant growth, because we essentially resample $s$ particles, and let each particle have $r$ descendants in a stratified manner. This idea is also closely related to the optimal sampling in the discrete space \citep{fearnhead2003line}. Figure~\ref{figure:discrepancy} compares the discrepancy set generated by stratified multiple-descendant growth and two other approaches.

\begin{figure}
\centering
\begin{subfigure}{.28\textwidth}
\centering
    \includegraphics[width=1\textwidth]{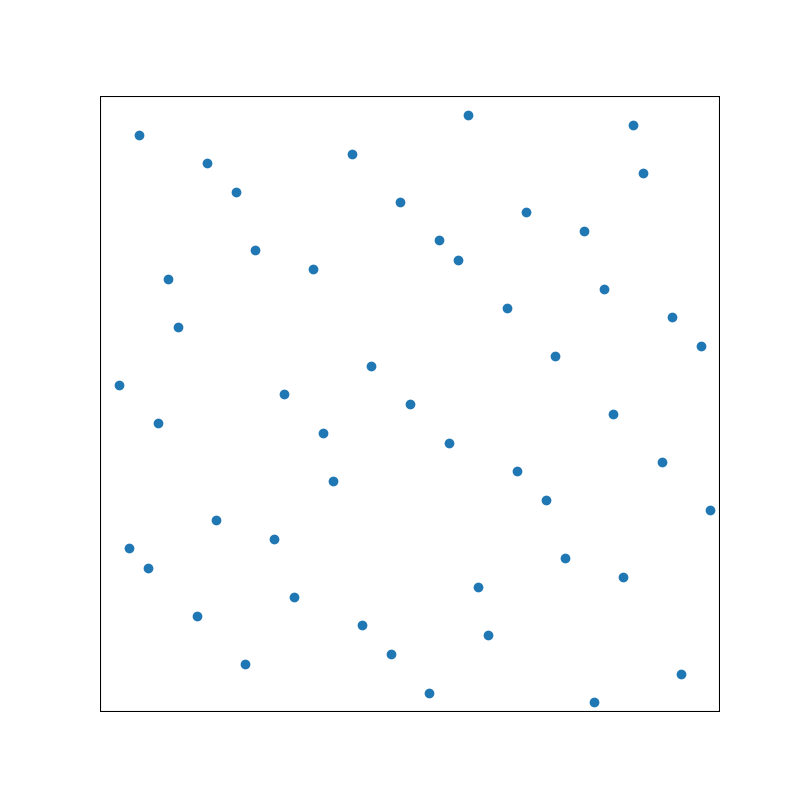}
    \caption{Sobol' sequence.}
    \centering
\end{subfigure}
\begin{subfigure}{.4\textwidth}
\centering
    \includegraphics[width=0.7\textwidth]{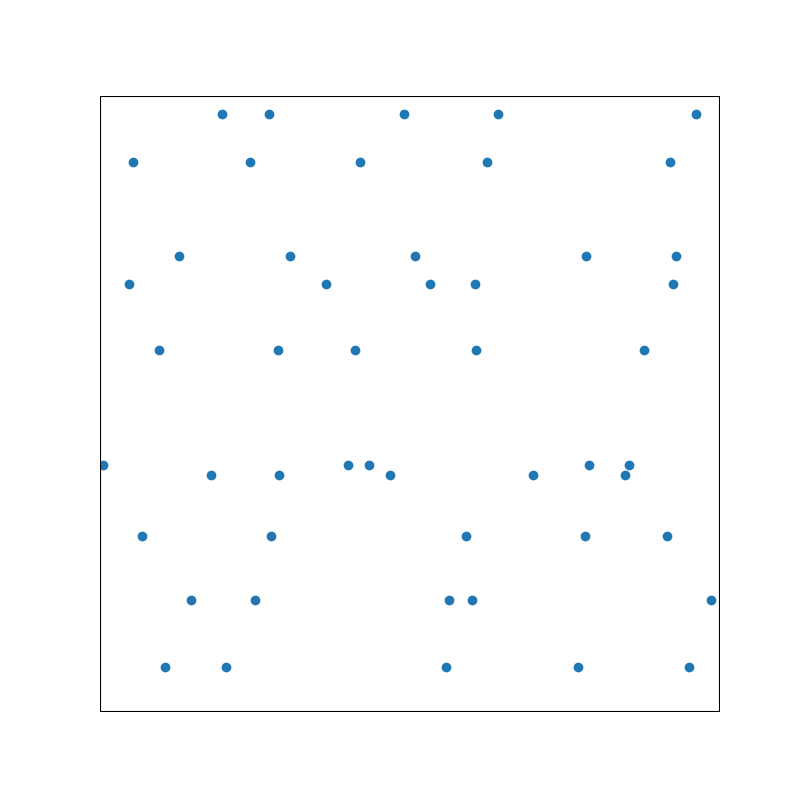}
    \caption{Stratified multiple-descendant growth.}
\end{subfigure}
\begin{subfigure}{.28\textwidth}
\centering
    \includegraphics[width=1\textwidth]{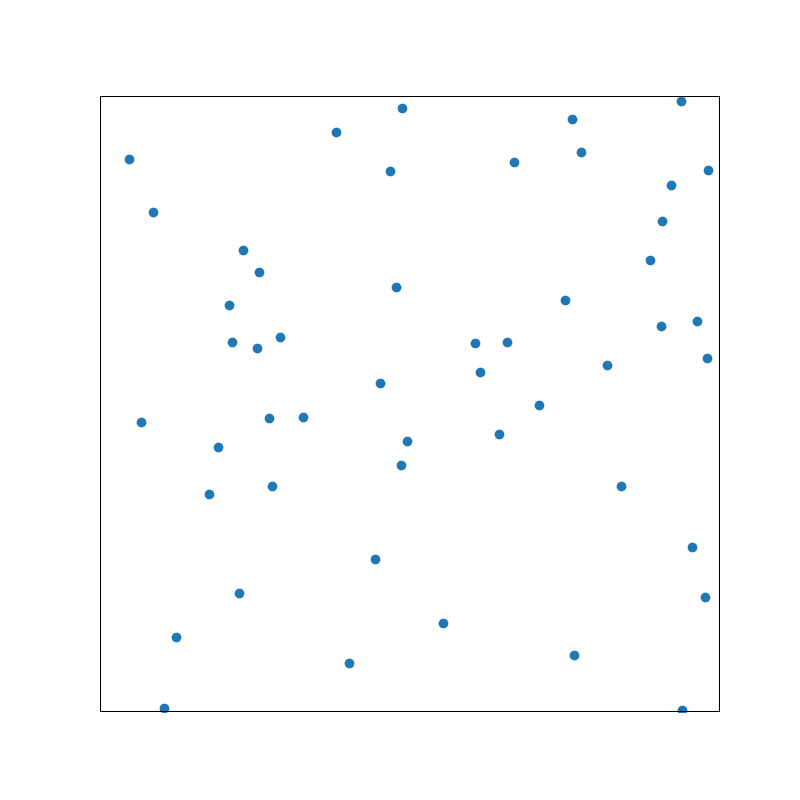}
    \caption{Independent sampling.}
\end{subfigure}
\caption{Comparison of low-discrepancy sets on $[0,1]^2$ ($n=50$, $k=10$, $r=5$).}
\label{figure:discrepancy}
\end{figure}

We focus on the mean square error of the estimation of any proper function $\phi$ in a state-space model, which is defined as
$$\text{MSE}_t(\phi) = \e \left[\frac{\sum_{j=1}^{n} W^{(t)}_{j} \phi(X^{(t)}_{j})}{\sum_{j=1}^{n} W^{(t)}_{j}}- \int \pi_{t}(x^{(1:t)})\phi(x^{(t)})dx^{(1:t)}\right]^2.$$
The mean square error can be decomposed into the squared bias and variance. The following theorem gives a bound for each one, respectively.
\begin{theorem}\label{thm: smg_convergence}
In a state-space model \eqref{equation:state-space}, we let $g_t(x^{(t)}\mid x^{(t-1)})=p_x((x^{(t)}\mid x^{(t-1)})$ and run sequential quasi-Monte Carlo with $U_\text{SMG}^{(t)}$. Assume that each $X^{(t)}$ falls in a compact set, assuming to be $\mathcal X=[0,1]^d$ without loss of generality.
Suppose $(X_{j}^{(t)}, W_{j}^{(t)})_{1\le j \le n}$ are the weighted samples at time $t$, where the number of multiple descendants $r=cn^{2/(d+4)}$ and particle dimension $d\ge 2$.
Assume that, for any t,
\begin{enumerate}
    \item[(i)] $a(v) = \pi_{t-1}\left(X\right)^{-1}g_t\left(v\mid X\right)^{-1} {\pi_{t}\left((X,v)\right) }$, $b(v) =  \pi_{t-1}\left((X, v)\right)^{-1} \pi_t\left(\left(X, v,u\right)\right)$, $c(v)=\Gamma_t(X,v)$, and $\Gamma_1(v)$ are bounded in $[-M,M]$ and $L$-Lipschitz.
    \item[(ii)] $\pi_{t-1}\left((X, v)\right)^{-1}\int_{\mathcal X} \pi_t\left(\left(X, v,u\right)\right)du$ is lower bounded by $\underline{e}>0$.
\end{enumerate}
Then, for any $L$-Lipschitz $\phi$ bounded in $[-M,M]$,
\begin{align*}
\left| \e\left[\frac{\sum_{j=1}^{n} W^{(t)}_{j} \phi(X^{(t)}_{j})}{\sum_{j=1}^{n} W^{(t)}_{j}}\right]- \int \pi_{t}(x^{(1:t)})\phi(x^{(t)})dx^{(1:t)}\right| &= O(n^{-\frac12-\frac{2}{d(d+4)}}), \\
\Var\left[ \frac{\sum_{j=1}^{n} W^{(t)}_{j} \phi(X^{(t)}_{j})}{\sum_{j=1}^{n} W^{(t)}_{j}}\right]&=O(n^{-1-\frac{4}{d(d+4)}})
\end{align*}
for all $t$, where the constants in $O$ depend only on $M$, $L$, $\underline{e}$ and $t$.
\end{theorem}
{\begin{remark}
There are different ways to map generally supported random vectors into $[0,1]^d$. Here we recommend the inverse transform method proposed in \citet{gerber2015sequential}. One significant advantage is that the spatial structure of the particles would be preserved to a large extent by the  method, which is important in the following Hilbert mapping. If the random vector is normally distributed, whitening the data is equivalent to the inverse transform method.
\end{remark}
}
{In dimension $d=2$, our simulations in a stochastic volatility model seem to suggest that the rate is rather tight. The results are shown in Figure~\ref{figure:rate} and the model details are included in Appendix~\ref{sec:sim-details}. We can see that the empirical slope gets closer to the slope $-4/3$ given by Theorem~\ref{thm: smg_convergence} as $n$ gets larger. 
\begin{figure}[H]
    \centering
\includegraphics[width = 0.7\textwidth]{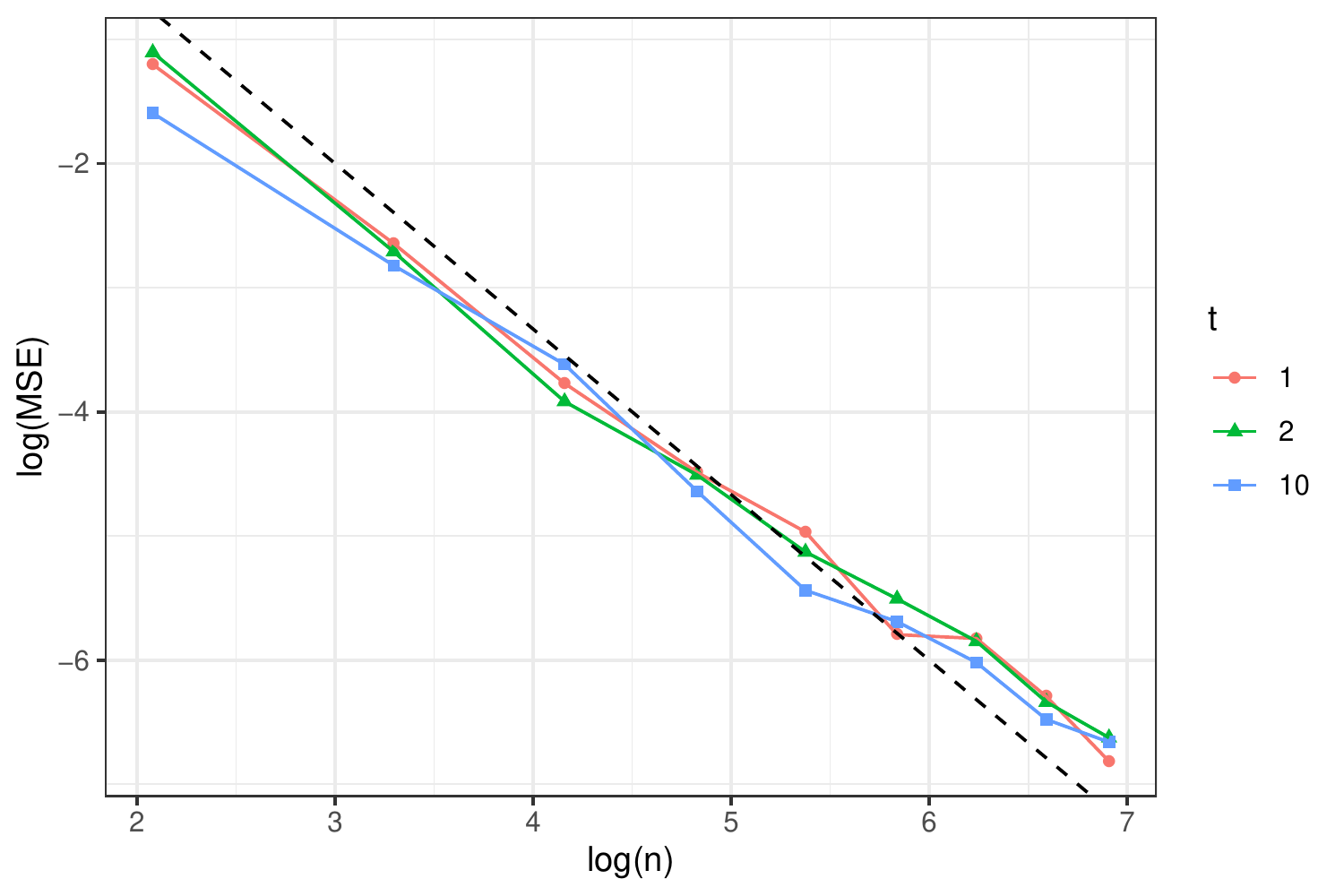}
    \caption{{The MSE versus the number of particles in logarithmic scales. The dashed line is a reference line with a slope of $-4/3$, the rate shown by our theory.}}
    \label{figure:rate}
\end{figure}}


\section{Discussion}
\label{sec:discussion}

We have discussed how one might improve the performance of SMC and SQMC via stratified sampling and multi-descendent growth.
The matrix resampling framework in Section~\ref{sec:resampling-matrix} can be generalized to allow  resampled particles to carry unequal weights (such as in the optimal resampling of \citet{fearnhead2003line}). Let $q_{1:m}$ satisfy $q_i\ge0$ and $\sum_{i=1}^mq_i=1$. We can resample according to a matrix $P=(p_{ij})_{m\times n}$ with non-negative entries where $\sum_{j=1}^np_{ij}=1$ and $\sum_{i=1}^mq_ip_{ij}=W_j$ by conditionally independent sampling:
    \begin{equation*}
    X_i^*\mid X,W\sim\text{Multinomial}(1,X,(p_{i1},p_{i2},\dots,p_{in})), i=1,2,\dots,m,
    \end{equation*}
    and then assigning $X_i^*$ the weight $q_i$. We focused on the case with $q_i=1/m$ in this article, but by choosing unequal $q_i$'s, one may further reduce the resampling variance at the cost of less balanced weights. It is unclear what an optimal trade-off might be.

When the resampled particles are not independent from each other conditional on the original particles, the resampling method cannot be represented as a resampling matrix. Systematic resampling \citep{carpenter1999improved} is such an example. All criteria mentioned in Section~\ref{sec:criteria} are also well-defined for non-matrix resampling. It would be interesting to study a broader class of resampling methods including some non-matrix resampling schemes.

Finally, it will be interesting to see if the tools in this paper can guide the choice of low-discrepancy sets or be generalized to analyze other existing low-discrepancy sets more commonly used in practice and show they achieve the same or better convergence rates. {In fact, \citet{gerber2015sequential} has shown by simulations that SQMC can significantly outperform SMC. We believe that SQMC with the Sobol sequence may have a better performance than multiple-descendent growth in practice based on our preliminary simulations.} It was conjectured that the optimal convergence rate {of SQMC} can reach $O(n^{-1-2/d})$ \citep{gerber2015sequential}.

\section*{Acknowledgement}
Y.~L. and W.~W. contributed equally and are listed in alphabetical order. Y.~L. is supported by China Scholarship Council. The research is partly supported by the Natural Science Foundation of China (Grant 11401338, KD PI), Beijing Academy of Artificial Intelligence (Supporting Grant, KD PI), and the National Science Foundation of USA (DMS-1712714 and DMS-1903139, JSL PI). The authors would like to thank Pierre Jacob for useful discussions on particle filters and optimal transport.

\bibliography{refs}
\newpage
\appendix
\section{Proofs}
\label{sec:proofs}

\begin{proof}[Proof of Lemma~\ref{lemma:unique_of_sm}]
Suppose $P = (p_{ij})_{m\times n}$ and $Q = (q_{ij})_{m \times n}$ are both eligible staircase matrices. If $p_{11} \neq q_{11}$, without loss of generality, assume $p_{11} < q_{11}$, then $\sum_{j=2}^n p_{1j} = r_1 - p_{11} > r_1 - q_{11} \geq 0$. By condition (1) in the definition of staircase matrix, $p_{12}>0$. This actually implies that $p_{i1} = 0$ for all $i > 1$. However, $p_{11}  = \sum_{i=1}^m p_{i1} = \sum_{i=1}^m q_{i1} \geq q_{11} > p_{11}$, which is a contradiction.

Then consider $p_{12}$ and $q_{12}$, suppose $0 \leq p_{12}<q_{12}$, then $\sum_{j=3}^n = p_{1j} = r_1 - p_{11} - p_{12} >r_1 - q_{11} - q_{12} \geq 0$. By condition (1) in the definition of staircase matrix, $p_{13}>0$. This implies that $p_{i2} =0$ for all $i>1$. Similarly, $p_{12} = \sum_{i=1}^m p_{i2} = \sum_{i=1}^m q_{i1} \geq q_{12} > p_{12}$, which is a contradiction.
\vspace{1em}
Similarly, we can prove that $p_{1j}= q_{1j}$ for each $j = 1,2, \cdots, n$. By induction, $P = Q$.
\end{proof}

\begin{proof}[Proof of Theorem~\ref{theorem: 1d-resampling}]
    First, we prove the following lemma.
\begin{lemma}\label{lemma:quad}
Suppose $P$ is an $m$ by $n$ matrix with $m,n>2$, $\sum_{i=1}^n p_{ij} >0$ for all $j$, and $\sum_{j=1}^n p_{ij} >0$ for all $i$, then in Definition~\ref{def:stair-mat}, $(2)$ implies $(1)$.
\end{lemma}

\begin{proof}[Proof of Lemma~\ref{lemma:quad}]
We consider the rows, and the same proof applies to the columns. Suppose $p_{ij_1} \neq 0$ and $p_{ij_2}\neq 0$, $j_1 < j_2$, for $j$ such that $j_1 < j < j_2$, if $p_{ij} = 0$, because $\sum_{s=1} p_{sj}>0$, there is a $k$ such that $p_{kj} >0$. If $k < i$, then $(k, j_1, i, j)$ is an ineligible quadruplet that contradicts (2). If $k > i$, then $(i, j, k, j_2)$ is an ineligible quadruplet that contradicts (2). 
\end{proof}
\begin{enumerate}
    \item[(i)] \text{Conditional variance.} 
    
    Suppose $P$ maximizes $t(P) = X^\T P^\T PX$ and $\sum_{j=1}^n p_{ij}X_j$ is ascending with respect to $i$ (note that permutation of rows in $P$ doesn't change the value of $X^\T P^\T PX$).
Consider a quadruplet $(i,j,k,l)$ such that $i<k$ and $j<l$. If $p_{il} >0$ and $p_{kj} > 0$, set $\alpha = \min\{p_{il}, p_{kj}\}>0$, then update the entries of $P$ as:
\begin{align*}
p_{ij} &\leftarrow p_{ij}+\alpha & p_{il} &\leftarrow p_{il}-\alpha\\
p_{kj} &\leftarrow p_{kj}-\alpha & p_{kl} &\leftarrow p_{kl} + \alpha
\end{align*}
We name the updated weight matrix as $P'$, then 
\begin{align*}
    t(P') - t(P) &= (\sum_{s=1}^n X_s p_{is} + \alpha(X_j-X_l))^2 + (\sum_{s = 1}^n X_s p_{ks} + \alpha(-X_j + X_l))^2 - \sum_{s=1}^n (X_s p_{is})^2 - \sum_{s=1}^n (X_s p_{ks})^2\\
    & = 2\alpha^2(X_j -X_l)^2 + 2\alpha(X_j -X_l)(\sum_{s=1}^n X_sp_{is} - \sum_{s=1}^n X_sp_{ks})>0,
\end{align*}
since $X_j < X_l$ and $\sum_{s=1}^n X_sp_{is} \leq \sum_{s=1}^n X_sp_{ks}$. This would contradict the fact that $P$ maximizes $t(P)$. Hence, by Lemma~\ref{lemma:quad}, $P$ is a staircase matrix.
    
    \item[(ii)] \text{Expected squared energy distance}.
    
    Note that the squared energy distance admits an explicit expression in one dimension. By some algebra, we find that Lemma~\ref{lemma:edist-equi} enables us to convert the problem of minimizing expected squared energy distance to a simpler problem.
\begin{lemma}
In the setting of Theorem~\ref{theorem: 1d-resampling}, the solution to the following optimization problems minimizes the expected squared energy distance:
\begin{equation*}
\argmax_{P\in\mathcal P_{m,W}}\sum_{k=1}^{n-1}\left[ (X_{k+1}-X_k)\sum_{i=1}^m\left(\sum_{j=1}^kp_{ij}\right)^2\right].
\end{equation*}
\label{lemma:edist-equi}
\end{lemma}

Back to the proof the theorem, let $P^\text{SR}_{m,W}=(p^*_{ij})$ be the ordered stratified resampling matrix. We will prove that for any $k$ and any $P=(p_{ij})\in\mathcal P_{m,W}$,
\begin{equation*}
\sum_{i=1}^m\left(\sum_{j=1}^kp^*_{ij}\right)^2\ge\sum_{i=1}^m\left(\sum_{j=1}^kp_{ij}\right)^2.
\end{equation*}
The result then follows from Lemma~\ref{lemma:edist-equi}. Since $\sum_{i=1}^m\left(\sum_{j=1}^kp_{ij}\right)=m\sum_{j=1}^kW_j$ and $0\le\sum_{j=1}^kp_{ij}\le1$, the sum of squares attains its maximum when $[m\sum_{j=1}^kW_j]$ of them are $1$, one of them is $m\sum_{j=1}^kW_j-[m\sum_{j=1}^kW_j]$, and the rest are $0$. It can be easily checked that $(p^*_{ij})$ satisfies this condition and thus solves the optimization problem.

\begin{proof}[Proof of Lemma~\ref{lemma:edist-equi}]
Let $d(\p,\tilde\p)=\int_{-\infty}^\infty(F_\p(x)-F_{\tilde\p}(x))^2\di{x}$, which is equal to half the squared energy distance \citep{szekely2003statistics}
\begin{equation*}
\e|X-Y|-\frac{\e|X-X'|+\e|Y-Y'|}{2},
\end{equation*}
with $X,X',Y,Y'$ independent, $X,X'$ coming from $\p$ and $Y,Y'$ coming from $\tilde\p$. Since the $X_j$'s are ordered as $X_1<X_2<\cdots<X_n$, we have
\begin{equation}
\begin{aligned}
\e[d(\p,\tilde\p)\mid X,W]&=\int_{-\infty}^\infty\e[(F_\p(x)-F_{\tilde\p}(x))^2\mid X,W]\di{x}\\
&=\int_{-\infty}^\infty(\e[F_{\tilde\p}(x)^2\mid X,W]-F_\p(x)^2)\di{x}.
\end{aligned}
\label{equation:energy-dist}
\end{equation}
Note that
\begin{equation*}
\begin{aligned}
\e[F_{\tilde\p}(x)^2\mid X,W]&=\frac1{m^2}\e[(\#\{i:\tilde{X}_i\le x\})^2\mid X,W]\\
&=\frac1{m^2}\sum_{i=1}^m\sum_{j=1}^kp_{ij}+\frac1{m^2}\sum_{i\ne l}\left(\sum_{j=1}^kp_{ij}\right)\left(\sum_{j=1}^kp_{l j}\right)\\
&=\underbrace{\frac1{m}\sum_{j=1}^kW_j}_{\text{constant}}+\frac1{m^2}\sum_{i\ne l}\left(\sum_{j=1}^kp_{ij}\right)\left(\sum_{j=1}^kp_{l j}\right), X_k\le x<X_{k+1}.
\end{aligned}
\end{equation*}
Minimizing equation~\eqref{equation:energy-dist} now becomes minimizing
\begin{equation*}
\begin{aligned}
&\sum_{k=1}^{n-1}(X_{k+1}-X_k)\left[\sum_{i\ne l}\left(\sum_{j=1}^kp_{ij}\right)\left(\sum_{j=1}^kp_{l j}\right)\right]\\
&=\sum_{k=1}^{n-1}(X_{k+1}-X_k)\left\{\left[\sum_{i=1}^m\left(\sum_{j=1}^kp_{ij}\right)\right]^2-\left[\sum_{i=1}^m\left(\sum_{j=1}^kp_{ij}\right)^2\right]\right\}\\
&=\sum_{k=1}^{n-1}(X_{k+1}-X_k)\left\{\left[m\left(\sum_{j=1}^kW_j\right)\right]^2-\left[\sum_{i=1}^m\left(\sum_{j=1}^kp_{ij}\right)^2\right]\right\},
\end{aligned}
\end{equation*}
which, after discarding constants, simplifies to maximizing
\begin{equation*}
\sum_{k=1}^{n-1}(X_{k+1}-X_k)\left[\sum_{i=1}^m\left(\sum_{j=1}^kp_{ij}\right)^2\right].
\end{equation*}
\end{proof}

\item[(iii)] \text{Earth mover's distance.} 

Let $t(P)=\sum_{i=1}^m \sum_{j=1}^n p_{ij}\ell(Y_i - X_j)$. Let $P$ be the matrix that minimizes $t(P)$. Consider a quadruplet $(i,j,k,l)$ such that $i<k$ and $j<l$. If $p_{il} >0$ and $p_{kj} > 0$, set $\alpha = \min\{p_{il}, p_{kj}\}>0$, then update the entries of $P$ as:
\begin{align*}
p_{ij} &\leftarrow p_{ij}+\alpha & p_{il} &\leftarrow p_{il}-\alpha\\
p_{kj} &\leftarrow p_{kj}-\alpha & p_{kl} &\leftarrow p_{kl} + \alpha
\end{align*}
We name the updated weight matrix as $P'$, then 
\begin{align*}
    t(P') - t(P) &= \alpha(\ell(Y_i-X_j)+\ell(Y_k-X_l)-\ell(Y_i-X_l)-\ell(Y_k-X_j)).
\end{align*}
Since $\ell$ is convex and
\begin{equation*}
\begin{aligned}
(Y_i-X_j)+(Y_k-X_l)&=(Y_i-X_l)+(Y_k-X_j)\\
|(Y_i-X_j)-(Y_k-X_l)|&<|(Y_i-X_l)-(Y_k-X_j)|\\
\end{aligned}
\end{equation*}
we have
\begin{equation*}
\ell(Y_i-X_j)+\ell(Y_k-X_l)<\ell(Y_i-X_l)+\ell(Y_k-X_j)),
\end{equation*}
so $t(P')< t(P)$. This would contradict the fact that $P$ is the minimizer, so such a quadruplet does not exist. By Lemma~\ref{lemma:quad}, the solution $P$ is a staircase matrix.

\end{enumerate}
\end{proof}

\begin{proof}[Proof of Theorem~\ref{theorem:variance_1d}]
Without loss of generality, suppose $X_1 < X_2 < \cdots < X_n$ and
$P$ is a staircase weight matrix corresponding to stratified resampling. Each {$\tilde X_i$} can only take values in $X_{il}, X_{il+1}, \cdots, X_{ir}$, with
{
\[
    X_{il} \le X_{ir} \text{  and  } X_{ir} \le X_{i+1,l},  \ \ \mbox{for $1\le i \le m-1$.}
\]
}

Hence, 
\begin{equation*}
\begin{split}
\Var\left[\dfrac{1}{m}\sum_{i=1}^m \phi(\tilde X_i)\mid{X},W\right]
&= \dfrac{1}{m^2}\sum_{i=1}^m \Var\left[\phi(\tilde X_i)\mid{X},W\right]\\
&\leq 
\dfrac{1}{m^2} \sum_{i=1}^m \dfrac{1}{4} \max_{x,y\in[X_{ir},X_{il}]}(\phi(x)-\phi(y))^2\\
&\quad\quad\quad\quad\quad\quad\quad\quad\quad\text{ (Popoviciu's inequality on variances{, Lemma~\ref{lemma:popoviciu})}}\\
&\leq 
\dfrac{1}{m^2} \sum_{i=1}^m \dfrac{1}{4}\max_{x,y\in[X_{ir},X_{il}]} L_\phi^2(x-y)^2
=\dfrac{L_\phi^2}{4m^2} \sum_{i=1}^m (X_{ir}-X_{il})^2\\
&\leq \dfrac{L_\phi^2}{4m^2}(X_n-X_1) {\color{red} \sum_{i=1}^m(X_{ir}-X_{il})}
 = \dfrac{L_\phi^2}{4m^2}(X_n- X_1)^2.
\end{split}
\end{equation*}
\end{proof}

{
\begin{lemma}[Popoviciu's inequaltity on variances, \citet{popoviciu1935equations}]
Let $M$ and $m$ be the upper bound and lower bound of a random variable $X$, respectively. Then,
$\Var(X) \le (M-m)^2/4.$
\label{lemma:popoviciu}
\end{lemma}
}

\begin{proof}[Proof of Theorem~\ref{theorem:variance_multi_dim}]
First note that $H(x)$ is H\"older continuous with exponent $1/d$, 
\begin{equation*}
\|H(x)- H(y)\| \leq 2\sqrt{d+3} |x- y|^{1/d}.
\end{equation*}
With Hilbert curve stratified sampling, $\tilde X_i$ can only take values in $X_{il}, X_{il+1}, \cdots, X_{ir}$, with
$$h(X_1) = h(X_{1l}) \leq \cdots \leq h(X_{i-1, r}) \leq  h(X_{il}) \leq h(X_{ir}) \leq h(X_{i+1, l}) \leq \cdots \leq h(X_{nr}) = h(X_n).$$
Note that
\begin{align*}
&\Var_P \left[\dfrac{1}{m}\sum_{i=1}^m \phi(\tilde X_i)\mid{X}\right]= \frac{1}{m^2} \sum_{i=1}^m \Var[\phi(\tilde X_i)\mid{X}]= \dfrac{1}{m^2} \sum_{i=1}^m \Var[\phi(H(h(\tilde{X}_i)))\mid{X}]\\
&\le \dfrac{1}{4m^2} \sum_{i=1}^m \left(\max_{x: h(x) \in [h(X_{il}), h(X_{ir})]} \phi(x) - \min_{x: h(x) \in [h(X_{il}), h(X_{ir})]} \phi(x)\right)^2\text{ (Lemma~\ref{lemma:popoviciu})}\\
&= \dfrac{1}{4m^2} \sum_{i=1}^m \left(\max_{y\in [h(X_{il}), h(X_{ir})]} \phi(H(y)) - \min_{y \in [h(X_{il}), h(X_{ir})]} \phi(H(y))\right)^2\\
&= \dfrac{1}{4m^2} \sum_{i=1}^m \max_{y_1,y_2\in[h(X_{il}), h(X_{ir})]}\|\phi(H(y_1))-\phi(H(y_2))\|^2\\
&\leq \dfrac{1}{4m^2} \sum_{i=1}^m \max_{y_1,y_2\in[h(X_{il}), h(X_{ir})]}L_\phi^2\|H(y_1)-H(y_2)\|^2\\
&\leq \dfrac{L_\phi^2}{4m^2} \sum_{i=1}^m\max_{y_1,y_2\in[h(X_{il}), h(X_{ir})]}4(d+3)|y_1-y_2|^{2/d}\\
&=\dfrac{(d+3)L_\phi^2}{m^2} \sum_{i=1}^m(h(X_{ir})-h(X_{il}))^{2/d}\\
&\leq \dfrac{(d+3)L_\phi^2}{m^2}\left[\sum_{i=1}^m ((h(X_{ir})-h(X_{il}))^{2/d})^{d/2}\right]^{2/d}m^{1-2/d}\text{ (H\"older inequality)}\\
&= \dfrac{(d+3)L_\phi^{2}m^{1-2/d}}{m^2}(h(X_m)-h(X_1))^{2/d}\le\dfrac{(d+3)L_\phi^2}{m^{1+2/d}}.
\end{align*}

\end{proof}

\begin{proof}[Proof of Proposition~\ref{proposition:optimality}]
We will prove that for all $n=2^{kd}$, where $k$ is a positive integer and $kd>3$, there exists $\phi\in\Phi_d$, $W$ and $X$ such that $$\Var_P(\frac1n\sum_{i=1}^n\phi(\tilde X_i)\mid X,W)\ge\frac1{27d}\frac{1}{n^{1+2/d}}.$$
Let
$$\mathcal{L}_k = \left\{0, \dfrac{1}{2^k}, \cdots, \dfrac{2^k-1}{2^k}\right\}^d$$ 
be an equally spaced grid of $[0,1]^d$. Let ${X} = (X_1, X_2, \cdots, X_{2^{dk}})$ be the sequence of points in $\mathcal{L}_k$ ordered by $o$.
Suppose
$${W}=(W_1,\cdots, W_{2^{dk}}) \propto (\underbrace{1, \cdots, 1}_{2^{kd - 1}}, \underbrace{2, \cdots, 2}_{2^{kd-1}}).$$
The stratified resampling matrix is
$${P} = \diag\{\underbrace{{P_1}, \cdots, {P_1}}_{(2^{dk-1}-2)/3}, {P_2}, \underbrace{{P_3}, \cdots, {P_3}}_{(2^{dk-1}-2)/3}\},$$
where
\[
{P_1} = \begin{pmatrix}
2/3 & 1/3 & \\
&1/3 & 2/3 
\end{pmatrix},  \ \ 
{P_2} = \begin{pmatrix}
2/3 & 1/3 & &\\
& 1/3 & 2/3 &\\
& & 2/3 & 1/3 \\
& & &1
\end{pmatrix},  \  \
{P_3} = \begin{pmatrix}
1 & & \\
1/3 & 2/3 &\\
& 2/3& 1/3 \\
& & 1
\end{pmatrix}.
\]

Let $\phi_{k}(X = (x_1, \cdots, x_d)) = x_k$ be the function that returns the $k$th coordinate, $k = 1,2,\cdots, d$. It is easy to see that $\phi_k$ is $1$-Lipschitz. We prove a simple lemma below.
\begin{lemma}
If $Z$ is a random variable defined by
$$Z = \begin{cases}
x, & \text{ with probability } 1/3, \\
y & \text{ with probability } 2/3,
\end{cases}$$
where $x$ and $y$ are distinct points in $\mathcal L_k$, then $\Var(\phi_k(Z)) \geq \dfrac{2^{-2k+1}}{9}$ for at least one $k\in\{1,2,\dots,d\}$.
\label{lemma:variance}
\end{lemma}
\begin{proof}[Proof of Lemma~\ref{lemma:variance}]
By direct calculation, $\Var(\phi_k(Z)) = \dfrac{2}{9}(x_k-y_k)^2$. Since $x\ne y$, at least one $k$ satisfies $|x_k-y_k|\ge2^{-k}$.
\end{proof}

Now the resampling variance is
\begin{eqnarray*}
\sum_{k=1}^d \dfrac{1}{m^2}\Var_P\left[\sum_{i=1}^{m} \phi_k(\tilde{X}_i)\mid {X}, {W}\right] 
&=& \dfrac{1}{m^2}\sum_{i=1}^{m}\sum_{k=1}^d \Var_P\left[\phi_k(\tilde{X}_i)\mid {X}, {W}\right] \\
&\ge& \dfrac{1}{m^2}\sum_{i=1}^{(2^{dk}-4)/3}\sum_{k=1}^d \Var_P\left[\phi_k(\tilde{X}_i)\mid {X}, {W}\right] \\
&\ge& \dfrac{1}{m^2}\sum_{i=1}^{(2^{dk}-4)/3} \dfrac{2^{-2k+1}}{9} 
= \dfrac{1}{2^{2dk}}\dfrac{2^{dk}-4}{3} \dfrac{2^{-2k+1}}{9} \\
&{\ge}& \dfrac{1}{2^{2dk}}\dfrac{2^{dk-1}}{3} \dfrac{2^{-2k+1}}{9}
=  \dfrac{1}{27}m^{-1 - 2/d} \text{ (when $dk\ge3$)} .
\end{eqnarray*}
Hence, there exists at least one $k \in \{1,2,\cdots, d\}$, such that 
$$\dfrac{1}{m^2}\Var_P\left[\sum_{i=1}^{m} \phi_k(\tilde{X}_i)\mid {X}, {W}\right] \geq \dfrac{1}{27d}\dfrac{1}{m^{1+2/d}}.$$
\end{proof}

\begin{proof}[Proof of Theorem~\ref{theorem:hc-wasserstein}]
We define a coupling between $Y\sim\p=\sum_{j=1}^nW_j\delta_{X_j}$ and $\tilde Y\sim\tilde\p=\sum_{i=1}^m\frac1m\delta_{\tilde{X}_i}$ by letting $(Y,\tilde Y)=(X_J,\tilde X_I)$, where $P(I=i,J=j)=p_{ij}/m$ and $p_{ij}$ is the $(i,j)$-entry of the Hilbert curve resampling matrix $P$. Recall that with Hilbert curve stratified sampling, $\tilde X_i$ can only take values in $X_{il}, X_{il+1}, \cdots, X_{ir}$, with
{\color{red}
\begin{equation*}
    h(X_{il}) \le h(X_{ir}) \text{  and  } h(X_{ir}) \le h(X_{i+1,l}), 
\end{equation*}
for $1\le i \le n$.
}
\begin{equation*}
\begin{aligned}
\e[\|Y-\tilde Y\|^p]&=\sum_{i=1}^m\sum_{j=1}^n\frac1mp_{ij}\|\tilde{X}_i-X_j\|^p\\
&\le\frac1m\sum_{i=1}^m\max_{z,z'\in[h(X_{il}),h(X_{ir})]}\|H(z)-H(z')\|^p\\
&\le\frac1m\sum_{i=1}^m(2\sqrt{d+3}(h(X_{ir})-h(X_{il}))^{1/d})^p\\
&\le\left\{
\begin{aligned}
&2^p(d+3)^{p/2}m^{-p/d},&\text{ if }p\le d,\\
&\frac{2^p(d+3)^{p/2}}{m},&\text{ if }p>d.\\
\end{aligned}\right.
\end{aligned}
\end{equation*}
Thus,
\begin{equation*}
W_p(\p^*,\p)\le\frac{2\sqrt{d+3}}{m^{1/\max(p,d)}},\quad a.s.
\end{equation*}
\end{proof}

\begin{proof}[Proof of  Theorem~\ref{thm: smg_convergence}]
We introduce some notations for presentational convenience to reflect the multiple-descendant nature. Let $\tilde X_{k}^{(t-1)}=X^{(t-1)}_{\sigma(U_k,W_{1:n}^{(t-1)})}$ corresponds to the resampled particle (recall the definition of $U_k$ for $U_\text{SMG}^{(t)}$) for $t\ge2$, and let
\begin{equation*}
X_{k\ell}^{(t)}=X_{k(s-1)+\ell}^{(t)}=\left\{
\begin{aligned}
&\Gamma_1(v_{k(s-1)+\ell}), &t=1,\\
&\Gamma_t(\tilde X^{(t-1)}_{k},v_{k(s-1)+\ell}), &2\le t\le T,
\end{aligned}
\right.
\end{equation*}
be the $\ell$th descendant of $\tilde X^{(t-1)}_k $. Similarly, $W_{k\ell}^{(t)}=W_{k(s-1)+\ell}^{(t)}$. See Figure~\ref{figure:multiple-des} for an illustration.

\begin{figure}
    \centering
    \includegraphics[width=.9\textwidth]{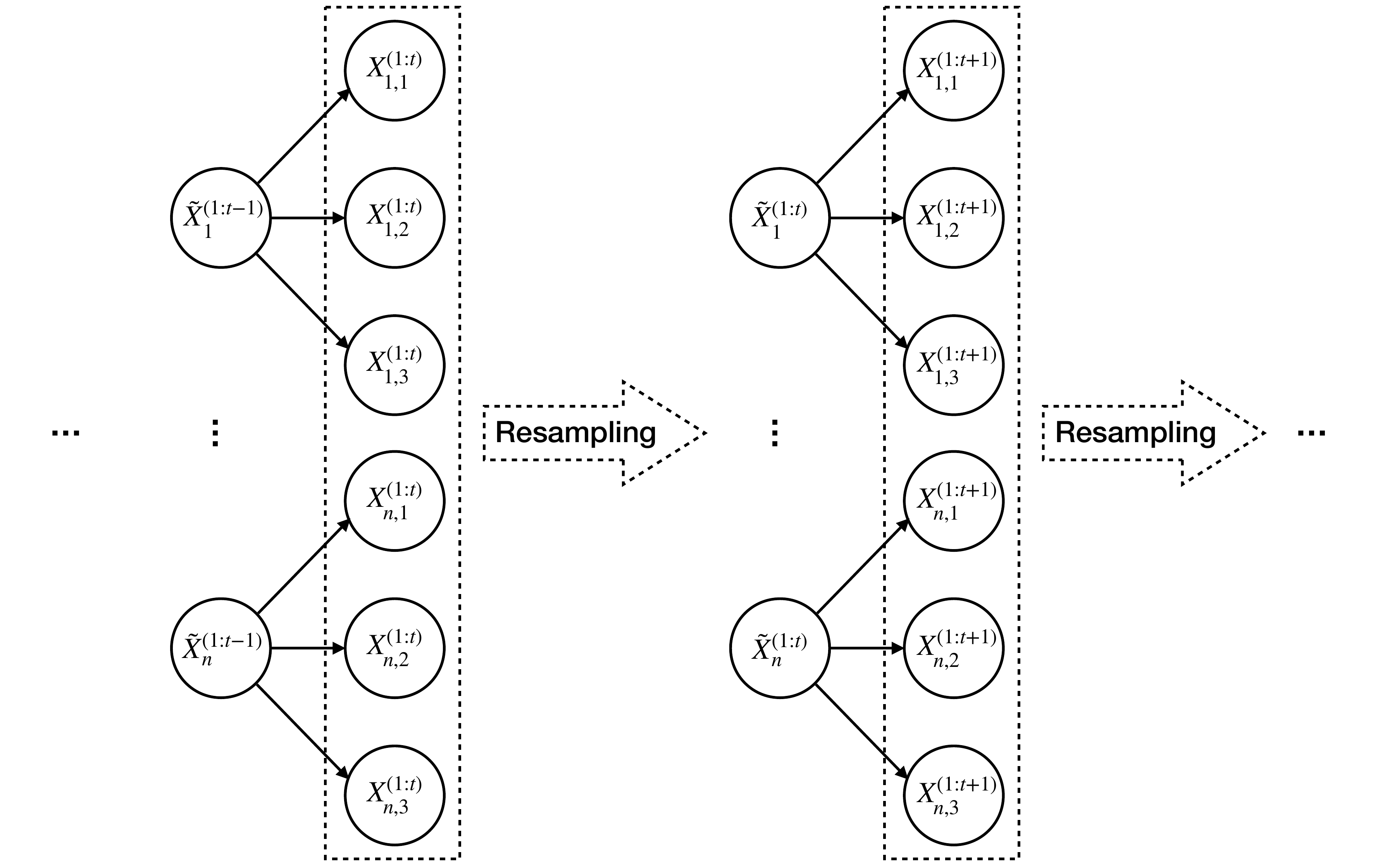}
    \caption{Illustration of multiple-descendant growth.}
    \label{figure:multiple-des}
\end{figure}
We then introduce two lemmas.

\begin{lemma}
Under the assumptions of Theorem~\ref{thm: smg_convergence},
\begin{equation*}
\Var\left[ \dfrac{1}{n}\sum_{k=1}^{s}\sum_{\ell=1}^{r} W^{(t)}_{k\ell} \phi(X^{(t)}_{k\ell})\mid X^{(1:t-1)}, W^{(t-1)} \right]=O(n^{-1-\frac4{d(d+4)}}).
\end{equation*}
\label{lemma:unnormalized-variance}
\end{lemma}

\begin{lemma}
Under the assumptions of Theorem~\ref{thm: smg_convergence},
\label{lemma:normalized-variance}
\begin{equation*}
\Var\left[ \frac{\dfrac{1}{n}\sum_{k=1}^{s}\sum_{\ell=1}^{r} W^{(t)}_{k\ell} \phi(X^{(t)}_{k\ell})}{\dfrac{1}{n}\sum_{k=1}^{s}\sum_{\ell=1}^{r} W^{(t)}_{k\ell}}\mid X^{(1:t-1)}, W^{(t-1)} \right]=O(n^{-1-\frac4{d(d+4)}}).
\end{equation*}
\end{lemma}

We prove by induction. For $t=1$,
\begin{equation}
    \begin{aligned}
    &\left( E\left[\frac{\sum_{k=1}^{s}\sum_{\ell=1}^{r} W^{(1)}_{ik} \phi(X^{(1)}_{ik})}{\sum_{k=1}^{s}\sum_{\ell=1}^{r} W^{(1)}_{ik}}- \int_{\mathcal X}\pi_1(x^{(1)})\phi(x^{(1)})d x^{(1)} \right]\right)^2 \\ 
    & = \left( E\left[\frac{\sum_{k=1}^{s}\sum_{\ell=1}^{r} W^{(1)}_{ik} \phi(X^{(1)}_{ik})}{\sum_{k=1}^{s}\sum_{\ell=1}^{r} W^{(1)}_{ik}}- \frac{\sum_{k=1}^{s}\sum_{\ell=1}^{r} W^{(1)}_{ik} \phi(X^{(1)}_{ik})}{sr} \right]\right)^2 \\ 
    & = \left( E\left[\frac{\sum_{k=1}^{s}\sum_{\ell=1}^{r} W^{(1)}_{ik} \phi(X^{(1)}_{ik})}{\sum_{k=1}^{s}\sum_{\ell=1}^{r} W^{(1)}_{ik}}\left(1 - \frac{\sum_{k=1}^{s}\sum_{\ell=1}^{r} W^{(1)}_{ik} }{sr}\right) \right]\right)^2 \\
    & \le E\left[\frac{\sum_{k=1}^{s}\sum_{\ell=1}^{r} W^{(1)}_{ik} \phi(X^{(1)}_{ik})}{\sum_{k=1}^{s}\sum_{\ell=1}^{r} W^{(1)}_{ik}}\left(1 - \frac{\sum_{k=1}^{s}\sum_{\ell=1}^{r} W^{(1)}_{ik} }{sr}\right) \right]^2\\
    & \le M^2\e\left[\left(1 - \frac{\sum_{k=1}^{s}\sum_{\ell=1}^{r} W^{(1)}_{ik} }{sr}\right)^2\right]
     = M^2\Var\left(\frac{\sum_{k=1}^{s}\sum_{\ell=1}^{r} W^{(1)}_{ik} }{sr}\right),
    \end{aligned}
    \label{equation:first-step-bias}
\end{equation}
since $n^{-1}\e\left(\sum_{k=1}^{s}\sum_{\ell=1}^{r} W^{(1)}_{k\ell}\right)=1$. The variance is $O(n^{-1-4/[d(d+4)]})$ by the same analysis as part A in the proof of Lemma~\ref{lemma:unnormalized-variance}.

Now suppose we have proved the cases from $1$ to $t-1$. Note that in state-space models, $h_t(X_{k\ell}^{(t-1)},x)= \pi_t\left(\left(\tilde X_{k\ell}^{(1:t-1)},x\right)\right)/\pi_{t-1}\left(X_{k\ell}^{(1:t-1)}\right)$ is only a function of $x$ and $X_{k\ell}^{(t-1)}$.
\begin{equation}
    \begin{aligned}
\label{equation:unnormalized-bias}
    & \left(\e \left[\frac{\sum_{k=1}^{s}\sum_{\ell=1}^{r} W^{(t)}_{k\ell} \phi(X^{(t)}_{k\ell})}{sr} - \int \pi_{t}(x^{(1:t)})\phi(x^{(t)})dx^{(1:t)}\right]\right)^2 \\\
    & = \Bigg( \dfrac{1}{s}\e \sum_{k=1}^s \e\left[\int_{\mathcal X} \frac{\pi_t\left(\left(\tilde X_k^{(1:t-1)},x\right)\right)}{\pi_{t-1}\left(\tilde{X}^{(1:t-1)}_{k}\right)}\phi(x)\di x\mid X_{k\ell}^{(1:t-1)}, W_{k\ell}^{(t-1)}\right]  \\
    & \quad - \int \frac{\pi_{t}(x^{(1:t)})}{\pi_{t-1}(x^{(1:t-1)})}\pi_{t-1}(x^{(1:t-1)})\phi(x^{(t)})dx^{(1:t)}\Bigg)^2 \\
    & = \left(\e \sum_{k=1}^s\sum_{\ell=1}^r \dfrac{W_{k\ell}^{(t-1)}}{\sum_{k=1}^s\sum_{\ell=1}^r W_{k\ell}^{(t-1)}}\int_{\mathcal X} h_t(X_{k\ell}^{(t-1)},x)\phi(x)\di x - \int \pi_{t-1}(x^{(1:t-1)})h_t(x^{(t-1)},x^{(t)})\phi(x^{(t)})dx^{(1:t)}\right)^2 \\
        & = \left(\e \left[\dfrac{ \sum_{k=1}^s\sum_{\ell=1}^rW_{k\ell}^{(t-1)}\tilde\phi_t(X_{k\ell}^{(t-1)})}{\sum_{k=1}^s\sum_{\ell=1}^r W_{k\ell}^{(t-1)}}\right]- \int \pi_{t-1}(x^{(1:t-1)})\tilde\phi_t(x^{(t-1)})dx^{(1:t-1)}\right)^2=O(n^{-1-\frac4{d(d+4)}}),
    \end{aligned}
\end{equation}
by induction hypothesis, where
{
\begin{equation*}
\tilde\phi_t(x)=\int_{\mathcal X}h_t(x,u)\phi(u)\di u
\end{equation*}
is bounded and Lipschitz since $\phi$ is bounded and $h_t(\cdot,u)$ is bounded and uniformly Lipschitz by assumption.
}

Now we analyze the difference bewteen normalized estimate and unnormalized estimate.
\begin{equation}
    \begin{aligned}
    &\left( E\left[\frac{\sum_{k=1}^{s}\sum_{\ell=1}^{r} W^{(t)}_{k\ell} \phi(X^{(t)}_{k\ell})}{\sum_{k=1}^{s}\sum_{\ell=1}^{r} W^{(t)}_{k\ell}}- \frac{\sum_{k=1}^{s}\sum_{\ell=1}^{r} W^{(t)}_{k\ell} \phi(X^{(t)}_{k\ell})}{sr} \right]\right)^2 \\ 
    & = \left( E\left[\frac{\sum_{k=1}^{s}\sum_{\ell=1}^{r} W^{(t)}_{k\ell} \phi(X^{(t)}_{k\ell})}{\sum_{k=1}^{s}\sum_{\ell=1}^{r} W^{(t)}_{k\ell}}\left(1 - \frac{\sum_{k=1}^{s}\sum_{\ell=1}^{r} W^{(t)}_{k\ell} }{sr}\right) \right]\right)^2 \\
    & \le E\left[\frac{\sum_{k=1}^{s}\sum_{\ell=1}^{r} W^{(t)}_{k\ell} \phi(X^{(t)}_{k\ell})}{\sum_{k=1}^{s}\sum_{\ell=1}^{r} W^{(t)}_{k\ell}}\left(1 - \frac{\sum_{k=1}^{s}\sum_{\ell=1}^{r} W^{(t)}_{k\ell} }{sr}\right) \right]^2\\
    & \le M^2\e\left[\left(1 - \frac{\sum_{k=1}^{s}\sum_{\ell=1}^{r} W^{(t)}_{k\ell} }{sr}\right)^2\right]\\
    & = M^2\left(\e\left[1 - \frac{\sum_{k=1}^{s}\sum_{\ell=1}^{r} W^{(t)}_{k\ell} }{sr}\right]\right)^2+M^2\Var\left(\frac{\sum_{k=1}^{s}\sum_{\ell=1}^{r} W^{(t)}_{k\ell} }{sr}\right).
    \end{aligned}
    \label{equation:normalized-bias}
\end{equation}
The first term is $O(n^{-1-4/[d(d+4)]})$ by the same deduction as \eqref{equation:unnormalized-bias}. For the second term,
\begin{equation*}
\begin{aligned}
\Var\left(\frac{\sum_{k=1}^{s}\sum_{\ell=1}^{r} W^{(t)}_{k\ell} }{sr}\right)&=\e\left[\Var\left(\frac{\sum_{k=1}^{s}\sum_{\ell=1}^{r} W^{(t)}_{k\ell} }{sr}\mid X^{(1:t-1)},W^{(1:t-1)}\right)\right]\\
&\quad+\Var\left(\e\left[\frac{\sum_{k=1}^{s}\sum_{\ell=1}^{r} W^{(t)}_{k\ell} }{sr}\mid X^{(1:t-1)},W^{(1:t-1)}\right]\right)\\
&=O(n^{-1-\frac4{d(d+4)}})\text{ (Lemma~\ref{lemma:unnormalized-variance})}\\
&\quad+\Var\left(\dfrac{ \sum_{i=1}^n\sum_{k=1}^rW_{ik}^{(t-1)}\tilde{1}_t(X_{ik}^{(t-1)})}{\sum_{i=1}^n\sum_{k=1}^r W_{ik}^{(t-1)}}\right),
\end{aligned}
\end{equation*}
where
\begin{equation*}
\tilde1_t(x)=\int_{\mathcal X}h_t(x,u)\di u.
\end{equation*}
Since $h_t(\cdot,u)$ is bounded and uniformly Lipschitz, $\tilde1_t$ is bounded and Lipschitz. By induction hypothesis, the variance term is $O(n^{-1-4/[d(d+4)]})$. Now we have
\begin{align*}
    &\left( E\left[\frac{\sum_{k=1}^{s}\sum_{\ell=1}^{r} W^{(t)}_{k\ell} \phi(X^{(t)}_{k\ell})}{\sum_{k=1}^{s}\sum_{\ell=1}^{r} W^{(t)}_{k\ell}}- \int \pi_{t}(x^{(1:t)})\phi(x^{(t)})dx^{(1:t)}\right]\right)^2 \\
    & = \Bigg( E\left[\frac{\sum_{k=1}^{s}\sum_{\ell=1}^{r} W^{(t)}_{k\ell} \phi(X^{(t)}_{k\ell})}{\sum_{k=1}^{s}\sum_{\ell=1}^{r} W^{(t)}_{k\ell}}- \frac{\sum_{k=1}^{s}\sum_{\ell=1}^{r} W^{(t)}_{k\ell} \phi(X^{(t)}_{k\ell})}{sr} \right] \\ 
    &  \quad + \left[\frac{\sum_{k=1}^{s}\sum_{\ell=1}^{r} W^{(t)}_{k\ell} \phi(X^{(t)}_{k\ell})}{sr} - \int \pi_{t}(x^{(1:t)})\phi(x^{(t)})dx^{(1:t)}\right]\Bigg)^2 \\
    & \leq 2\left( E\left[\frac{\sum_{k=1}^{s}\sum_{\ell=1}^{r} W^{(t)}_{k\ell} \phi(X^{(t)}_{k\ell})}{\sum_{k=1}^{s}\sum_{\ell=1}^{r} W^{(t)}_{k\ell}}- \frac{\sum_{k=1}^{s}\sum_{\ell=1}^{r} W^{(t)}_{k\ell} \phi(X^{(t)}_{k\ell})}{sr} \right]\right)^2 \\ 
    &  \quad + 2\left(\e \left[\frac{\sum_{k=1}^{s}\sum_{\ell=1}^{r} W^{(t)}_{k\ell} \phi(X^{(t)}_{k\ell})}{sr} - \int \pi_{t}(x^{(1:t)})\phi(x^{(t)})dx^{(1:t)}\right]\right)^2 =O(n^{-1-\frac4{d(d+4)}}),
\end{align*}
by \eqref{equation:unnormalized-bias} and \eqref{equation:normalized-bias}. This completes the induction hypothesis for the bias at step $t$,

For the variance at step $t$,
\begin{equation*}
\begin{aligned}
\Var\left( \frac{\sum_{k=1}^{s}\sum_{\ell=1}^{r} W^{(t)}_{k\ell} \phi(X^{(t)}_{k\ell})}{\sum_{k=1}^{s}\sum_{\ell=1}^{r} W^{(t)}_{k\ell}}\right)&=\e\left[\Var\left( \frac{\sum_{k=1}^{s}\sum_{\ell=1}^{r} W^{(t)}_{k\ell} \phi(X^{(t)}_{k\ell})}{\sum_{k=1}^{s}\sum_{\ell=1}^{r} W^{(t)}_{k\ell}}\mid X^{(1:t-1)},W^{(1:t-1)}\right)\right]\\
&\quad+\Var\left(\e\left( \frac{\sum_{k=1}^{s}\sum_{\ell=1}^{r} W^{(t)}_{k\ell} \phi(X^{(t)}_{k\ell})}{\sum_{k=1}^{s}\sum_{\ell=1}^{r} W^{(t)}_{k\ell}}\mid X^{(1:t-1)},W^{(1:t-1)}\right]\right).
\end{aligned}
\end{equation*}
The first term is $O(n^{-1-\frac4{d(d+4)}})$ by Lemma~\ref{lemma:normalized-variance}. For the second term,
\begin{equation*}
\begin{aligned}
&\Var\left(\e\left( \frac{\sum_{k=1}^{s}\sum_{\ell=1}^{r} W^{(t)}_{k\ell} \phi(X^{(t)}_{k\ell})}{\sum_{k=1}^{s}\sum_{\ell=1}^{r} W^{(t)}_{k\ell}}\mid X^{(1:t-1)},W^{(1:t-1)}\right]\right)\\
&\le2\Var\left(\e\left( \frac{\sum_{k=1}^{s}\sum_{\ell=1}^{r} W^{(t)}_{k\ell} \phi(X^{(t)}_{k\ell})}{\sum_{k=1}^{s}\sum_{\ell=1}^{r} W^{(t)}_{k\ell}}-\frac{\sum_{k=1}^{s}\sum_{\ell=1}^{r} W^{(t)}_{k\ell} \phi(X^{(t)}_{k\ell})}{sr}\mid X^{(1:t-1)},W^{(1:t-1)}\right]\right)\\
&\quad+2\Var\left(\e\left(\frac{\sum_{k=1}^{s}\sum_{\ell=1}^{r} W^{(t)}_{k\ell} \phi(X^{(t)}_{k\ell})}{sr}\mid X^{(1:t-1)},W^{(1:t-1)}\right]\right)\\
&\le2\e\left[\left( \frac{\sum_{k=1}^{s}\sum_{\ell=1}^{r} W^{(t)}_{k\ell} \phi(X^{(t)}_{k\ell})}{\sum_{k=1}^{s}\sum_{\ell=1}^{r} W^{(t)}_{k\ell}}-\frac{\sum_{k=1}^{s}\sum_{\ell=1}^{r} W^{(t)}_{k\ell} \phi(X^{(t)}_{k\ell})}{sr}\right)^2\right]\\
&\quad+2\Var\left(\dfrac{ \sum_{i=1}^n\sum_{k=1}^rW_{ik}^{(t-1)}\tilde\phi_t(X_{ik}^{(t-1)})}{\sum_{i=1}^n\sum_{k=1}^r W_{ik}^{(t-1)}}\right)\\
&=O(n^{-1-\frac4{d(d+4)}})\text{ (derivation in \eqref{equation:normalized-bias})}\\
&\quad+ O(n^{-1-\frac4{d(d+4)}})\text{ (induction hypothesis)}.
\end{aligned}
\end{equation*}
This proves the induction hypothesis for the variance at step $t$. 
\end{proof}

\begin{proof}[Proof of Lemma~\ref{lemma:unnormalized-variance}]
We omit the superscript $(t-1)$ and $(1:t-1)$. We can decompose the conditional variance into two parts:
\begin{multline*}
    \Var\left[ \dfrac{1}{n}\sum_{k=1}^{s}\sum_{\ell=1}^{r} W^{(t)}_{k\ell} \phi(X^{(t)}_{k\ell})\mid X, W \right]= \underbrace{ E \left[\Var\left( \dfrac{1}{n}\sum_{k=1}^{s}\sum_{\ell=1}^{r}  W^{(t)}_{k\ell} \phi( X^{(t)}_{k\ell})\mid \tilde{X}_i, \tilde{W}_i \right)\mid   X, W\right]}_{A}\\
    +  \underbrace{\Var \left[ E\left( \dfrac{1}{n}\sum_{k=1}^{s}\sum_{\ell=1}^{r} W^{(t)}_{k\ell} \phi( X^{(t)}_{k\ell})\mid \tilde{X}_i, \tilde{W}_i \right)\mid   X, W\right]}_B.
\end{multline*}
To make the computation easy to read, we first analyze $A$ and $B$ separately.

Let $l_k(x) =  \dfrac{\pi_{t}\left((\tilde X_k,x)\right) \phi(x)}{\pi_{t-1}\left(\tilde X_k\right)g\left(x\mid \tilde{X}_k\right)}$,  which is Lipschitz by assumption. Suppose the Lipschitz constant is $L_k$, which, for example, can be $2ML$.

\begin{align*}
    &A =  E \left[\Var\left( \dfrac{1}{n}\sum_{k=1}^{s}\sum_{\ell=1}^{r}  W^{(t)}_{k\ell} \phi( X^{(t)}_{k\ell})\mid \tilde{X}_i, \tilde{W}_i \right)\mid   X, W\right] \\
    &= \dfrac{1}{n^2} \sum_{k=1}^{s}\sum_{\ell=1}^{r}  E \left[\Var\left( W^{(t)}_{k\ell} \phi( X^{(t)}_{k\ell})\mid \tilde{X}_i, \tilde{W}_i \right)\mid   X, W\right] \\
    & \leq \dfrac{1}{4n^2} \sum_{k=1}^{s}\sum_{\ell=1}^{r}  E \left[ \max_{x, y \in \Gamma_{t}(\tilde X_k,H([(\ell-1)/r,\ell/r]))} \left(l_k(x) - l_k(y)\right)^2 \mid   X, W\right]\text{ (Popoviciu's inequality on variance)} \\
    & \leq \dfrac{1}{4n^2} \sum_{k=1}^{s}\sum_{\ell=1}^{r} E \left[ \max_{x, y \in \Gamma_{t}(\tilde X_k,H([(\ell-1)/r,\ell/r]))} L_k^2||x-y||^2 \mid   X, W\right] \\
    & \leq \dfrac{1}{4n^2} \sum_{k=1}^{s}\sum_{\ell=1}^{r} E \left[ \max_{x, y \in H([(\ell-1)/r,\ell/r])} L^2L_k^2||x-y||^2 \mid   X, W\right] \\
    & \leq \dfrac{1}{4n^2} \sum_{k=1}^{s}\sum_{\ell=1}^{r}  E \left[ 4(d+3)L^2L_k^2(1/r)^{2/d} \mid   X, W\right] \text{ (H\"older continuity)}\\
    & = \dfrac{d+3}{n^2}  \sum_{k=1}^{s} L^2L_k^2 r^{1-2/d}= \dfrac{(d+3)L^2L_k^2 r^{-2/d}}{n}=O(n^{-1-\frac4{d(d+4)}}).
\end{align*}

For part $B$, we have
\begin{align*}
    &B = \Var \left[ E\left( \dfrac{1}{n}\sum_{k=1}^{s}\sum_{\ell=1}^{r} W^{(t)}_{k\ell} \phi( X^{(t)}_{k\ell})\mid \tilde{X}_k, \tilde{W}_k \right)\mid   X, W\right] \\
    & = \Var \left[\dfrac{r}{n}\sum_{k=1}^{s}  \int_{\mathcal X} \frac{\pi_t\left(\left(\tilde X_k^{(1:t-1)},x\right)\right)}{\pi_{t-1}\left(\tilde{X}^{(1:t-1)}_{k}\right)}\phi(x) dx\mid   X, W\right]
\end{align*}
Let $$f_k(x) =  \int_{\mathcal X} \frac{\pi_t\left(\left(\tilde X_k^{(1:t-2)}, x,u\right)\right)}{\pi_{t-1}\left(\tilde{X}^{(1:t-2)}_{k}, x\right)}\phi(u) du,$$ which is Lipschitz by assumption. Suppose the Lipschitz constant is $L_B$ for all $k$, which, for example, can be $2ML$. Then by Theorem~\ref{theorem:variance_multi_dim} (Theorem~\ref{theorem:variance_multi_dim} requires the functions to be the same, but actually it can be seen that the proof still applies as long as all the functions have the same Lipschitz constant),
$$B \leq \dfrac{(d+3)L^2_B}{(n/r)^{1+2/d}} = \dfrac{(d+3)L^2_B r^{1+2/d}}{n^{1+2/d}}=O(n^{-1-\frac4{d(d+4)}}).$$
Combining $A$ and $B$ proves the claim.
\end{proof}

\begin{proof}[Proof of Lemma~\ref{lemma:normalized-variance}]
Let $e(X,W)=\e[ \dfrac{1}{n}\sum_{k=1}^{s}\sum_{\ell=1}^{r} W^{(t)}_{k\ell}\mid X,W]$.
\begin{multline*}
\Var\left[ \frac{\dfrac{1}{n}\sum_{k=1}^{s}\sum_{\ell=1}^{r} W^{(t)}_{k\ell} \phi(X^{(t)}_{k\ell})}{\dfrac{1}{n}\sum_{k=1}^{s}\sum_{\ell=1}^{r} W^{(t)}_{k\ell}}\mid X, W \right]\le2\Var\left[\frac{ \dfrac{1}{n}\sum_{k=1}^{s}\sum_{\ell=1}^{r} W^{(t)}_{k\ell} \phi(X^{(t)}_{k\ell})}{e(X,W)}\mid X, W \right]\\
+2\Var\left[ \frac{\dfrac{1}{n}\sum_{k=1}^{s}\sum_{\ell=1}^{r} W^{(t)}_{k\ell} \phi(X^{(t)}_{k\ell})}{\dfrac{1}{n}\sum_{k=1}^{s}\sum_{\ell=1}^{r} W^{(t)}_{k\ell}}\left(1-\frac{\dfrac{1}{n}\sum_{k=1}^{s}\sum_{\ell=1}^{r} W^{(t)}_{k\ell}}{e(X,W)}\right)\mid X, W \right].
\end{multline*}
On the right hand side, the first term is $O(n^{-1-4/[d(d+4)]})$ by Lemma~\ref{lemma:unnormalized-variance}; for the second term,
\begin{align*}
&\Var\left[ \frac{\dfrac{1}{n}\sum_{k=1}^{s}\sum_{\ell=1}^{r} W^{(t)}_{k\ell} \phi(X^{(t)}_{k\ell})}{\dfrac{1}{n}\sum_{k=1}^{s}\sum_{\ell=1}^{r} W^{(t)}_{k\ell}}\left(1-\frac{\dfrac{1}{n}\sum_{k=1}^{s}\sum_{\ell=1}^{r} W^{(t)}_{k\ell}}{e(X,W)}\right)\mid X, W \right]\\
&\le\e\left[\left(\frac{\dfrac{1}{n}\sum_{k=1}^{s}\sum_{\ell=1}^{r} W^{(t)}_{k\ell} \phi(X^{(t)}_{k\ell})}{\dfrac{1}{n}\sum_{k=1}^{s}\sum_{\ell=1}^{r} W^{(t)}_{k\ell}}\right)^2\left(1-\frac{\dfrac{1}{n}\sum_{k=1}^{s}\sum_{\ell=1}^{r} W^{(t)}_{k\ell}}{e(X,W)}\right)^2\mid X, W \right]\\
&\le M^2\e\left[\left(1-\frac{\dfrac{1}{n}\sum_{k=1}^{s}\sum_{\ell=1}^{r} W^{(t)}_{k\ell}}{e(X,W)}\right)^2\mid X, W \right]\\
&= \frac{C_\phi^2}{e(X,W)^2}\Var\left[\dfrac{1}{n}\sum_{k=1}^{s}\sum_{\ell=1}^{r} W^{(t)}_{k\ell}\mid X, W \right]\\
&\le \frac{C_\phi^2}{\underline{e}^2}\Var\left[\dfrac{1}{n}\sum_{k=1}^{s}\sum_{\ell=1}^{r} W^{(t)}_{k\ell}\mid X, W \right]=O(n^{-1-\frac4{d(d+4)}})
\end{align*}
by taking $\phi$ to be the constant function $1$ in Lemma~\ref{lemma:unnormalized-variance}.
\end{proof}
\section{Details of the Stochastic Volatility Model}
\label{sec:sim-details}

The multidimensional stochastic volatility model is
\begin{align*}
    X^{(t)}\mid X^{(1:t-1)}&\sim\mathcal N\left(\alpha X^{(t-1)},\Sigma\right),\\
Y^{(t)}\mid X^{(1:t)},Y^{(1:t-1)}&\sim\mathcal N\left(0,\beta^2\diag(\exp(X^{(t)}))\right),
\end{align*}
for $t = 1,2,\cdots, T$.

Here we set the dimension of this model as 2, with $\alpha = 0.7$, $\beta = 1$, $T = 10$, and $X^{(0)}\sim \mathcal{N}\left(0, \Sigma\right)$, where $\Sigma = \begin{pmatrix}
1 & 0.8 \\
0.8 & 1
\end{pmatrix}$. 

We first generate $(X^{(t)}, Y^{(t)})_{t=1}^T$ from the above model, and then run the sequential quasi-Monte Carlo with $U^{(t)}_\text{SMG}$ at each time $t$ for 100 times.
Within each run, the number of multiple descendants ranges from 2 to 10, with the total number of the particles $n = r^3$.

\end{document}